\documentclass[11pt]{article}
\usepackage{graphicx}
\usepackage{amsfonts}
\usepackage{amsmath}
\usepackage{amssymb}
\usepackage{url}
\usepackage{fancyhdr}
\usepackage{indentfirst}
\usepackage{enumerate}
\usepackage{dsfont,footmisc}
\usepackage{comment}

\usepackage{setspace}
\usepackage[misc]{ifsym}

\usepackage{amsthm}
\usepackage{color}
\usepackage{natbib}
\usepackage[british]{babel}
\usepackage[colorlinks=true,citecolor=blue]{hyperref}
\usepackage[right,pagewise,mathlines]{lineno}

\usepackage{subcaption}
\usepackage{threeparttable}
\usepackage{multirow}
\usepackage{adjustbox}
\usepackage{booktabs}
\usepackage{tikz}

\usetikzlibrary{arrows.meta, positioning, calc} 
\usetikzlibrary{decorations.pathreplacing}
\usepackage{tikz-qtree}

\usetikzlibrary{arrows.meta,positioning}

\def\d{\mathrm{d}}

\newcommand{\calP}{\mathcal{P}}
\newcommand{\calF}{\mathcal{F}}
\newcommand{\calQ}{\mathcal{Q}}
\newcommand{\calX}{\mathcal{X}}

\newcommand{\R}{\mathbb{R}}

\renewcommand{\geq}{\geqslant}
\renewcommand{\leq}{\leqslant}
\renewcommand{\epsilon}{\varepsilon}

\theoremstyle{plain}
\newtheorem{theorem}{Theorem}
\newtheorem{lemma}{Lemma}
\newtheorem{proposition}{Proposition}

\theoremstyle{definition}
\newtheorem{definition}{Definition}
\newtheorem{example}{Example}

\theoremstyle{remark}
\newtheorem{remark}{Remark}

\theoremstyle{definition}

\renewcommand{\cite}{\citet}
\renewcommand{\cdots}{\dots}

\usepackage{xcolor}

\usepackage{geometry}
\begin{document} 

\title{Weighted Generalized Risk Measure and  
Risk Quadrangle: Characterization, Optimization and Application}

\author{
Yang Liu\thanks{\scriptsize School of Science and Engineering, The Chinese University of Hong Kong (Shenzhen), China. Email: \texttt{yangliu16@cuhk.edu.cn}} 
\and
Yunran Wei\thanks{\scriptsize School of Mathematics and Statistics, Carleton University, Canada. Email: \texttt{yunran.wei@carleton.ca}} 
\and
Xintao Ye\thanks{\scriptsize Corresponding Author. School of Mathematics, Shandong University, China. Email: \texttt{xintaoye2@gmail.com}} 
}

\date{}

\maketitle

\begin{abstract}

Various financial market scenarios may cause heterogeneous risk assessments among analysts, which motivates the usage of the Generalized Risk Measure in \cite{Fadina2024}. Effectively synthesizing these diverse assessments avoids over-relying on a single, potentially flawed or conservative forecast and promotes more robust decision-making. Motivated by this, we establish analytical characterizations of the Weighted Generalized Risk Measure (WGRM) under both discrete and continuous settings. Building upon the WGRM, we incorporate the Fundamental Risk Quadrangle (FRQ) in \cite{Rockafellar2013} into the Weighted Risk Quadrangle (WRQ) and show that the intrinsic relationships among risk, deviation, regret, error, and statistics in FRQ are preserved under weighted aggregation across scenarios. Moreover, we demonstrate that certain complex risk optimization problems under the WGRM can be reformulated as tractable linear programs through the WRQ structure, thus ensuring computational feasibility. Finally, the WGRM and WRQ framework is applied to empirical analyses using constituents of the NASDAQ 100 and S\&P 500 indices across recession and expansion regimes, which validates that WGRM-based portfolios exhibit superior risk-adjusted performance and enhanced downside resilience and  effectively mitigate losses arising from erroneous single-scenario judgments.



\end{abstract}

\section{Introduction}
\label{sec:intro}
\subsection{Background and Motivation}
The global financial crisis of 2008 revealed a fundamental weakness in modern risk management: risk assessments that appeared conservative in normal times proved fragile when market regimes shifted. Many highly rated financial institutions collapsed not solely due to excessive risk-taking, but because of an over-reliance on a single or a narrow set of scenarios and the unchecked confidence in singular probabilistic models. A key lesson from this episode is that financial risk is inherently scenario-dependent. In adverse markets, different analysts often arrive at different risk assessments—even when applying the same formal risk measure—since they rely on distinct data sources, data process techniques and stress scenarios. In such settings, the central question for a department manager is no longer whether a risk measure is conservative enough, but how to avoid the catastrophic losses that can result from trusting any single, potentially flawed analyst's forecast in isolation \citep{Cont2010, Embrechts2015}.

As a response, practitioners increasingly rely on consulting and combining heterogeneous scenario analyses. Consequently, the traditional risk measures—which are typically characterized as functionals operating either on random variable spaces or, given law-invariance, on the spaces comprising their distribution functions (see \citealt{Artzner1999})—are no longer adequate. It is of significance and interest to switch from single-scenario traditional risk measures to multi-scenario settings; see \cite{Kou&Peng2016,Wang&Ziegel2021}. A recent framework in \cite{Fadina2024}, the generalized risk measure (GRM), accepts two inputs: the loss variable $X$ along with a collection $\calQ$ of admissible probability measures, which together provide a richer characterization of the underlying stochastic environment. 

Under that framework, for an arbitrary measurable space $(\Omega,\calF)$, we denote by $\calP$ the collection of all atomless probability measures on $\calF$ and, for simplicity, we assume it is a compact Polish space. Correspondingly, let $\calX$ denote the set of all random variables defined on this space. The power set of $\calP$ is denoted $2^{\calP}$. The formal definition of GRM in  \cite{Fadina2024} is as follows.

\begin{definition}
\label{def:GRM}
    A generalized risk measure is a mapping $\Psi$: $\calX \times 2^{\calP} \to (-\infty, \infty)$.
\end{definition}

The worst-case, coherent, and robust GRMs are characterized via different sets of axioms in \cite{Fadina2024}, where the worst-case one might be too conservative. Actually, the weighted aggregation form is worthy of exploration as it fully incorporates different analysts' expertise.
Motivated by this, we introduce the Weighted Generalized Risk Measure (WGRM), a principled framework designed to aggregate heterogeneous risk perspectives into a single, analytically tractable functional under both discrete and continuous settings to align with real-world risk management needs. Specifically, For any $\mathcal{Q} \subseteq \mathcal{P}$, we aim to find a probability measure $\mu_{\mathcal{Q}}$ that assigns weights to the elements of $\mathcal{Q}$, satisfying $\mu_{\mathcal{Q}}(\emptyset)=0$ and $\mu_{\mathcal{Q}}(\mathcal{Q})=1$, along with $0\leq \mu_{\mathcal{Q}}(P) \leq 1$ for all $P\in \mathcal{Q}$. We formally define WGRM as follows.

\begin{definition}
\label{def: WeightedGRM}
    The mapping $\Psi$ admits a weighted generalized risk measure (WGRM) representation if for every $\calQ \subseteq \calP$,$\,\exists\,\, \mu_{\calQ}$ such that
    \begin{equation}
    \text{(continuous)} \quad \quad    \Psi(X|\calQ)=\int_{\calQ} \Psi(X|P) \,\d\mu_{\calQ}(P),
        \quad
        X \in \calX,\ \calQ \in 2^{\calP}.
        \label{eq:weighted_grm}
    \end{equation}
    If $\calQ$ is a discrete set, then $\mu_{\calQ}$ reduces to discrete measure and $\Psi$ is equivalent as
    \begin{equation}
    \text{(discrete)} \quad \quad     \Psi(X|\calQ)=\sum_{P\in\calQ} \Psi(X|P) \cdot \mu_{\calQ}(P),
        \quad
        X \in \calX,\ \calQ \in 2^{\calP}.
    \end{equation}
\end{definition}

Subsequently, building on the WGRM model, we further extend the Fundamental Risk Quadrangle (FRQ), which is a framework combining optimization and estimation in risk management proposed by \cite{Rockafellar2013}. We use Figure \ref{fig:quadrangle} to briefly illustrate the FRQ with details postponed to Section \ref{sec:E_quadrangle}. By leveraging the weight structure derived from WGRM, we perform a weighted aggregation of the single-probability-measure risk quadrangle to develop the Weighted Risk Quadrangle (WRQ).

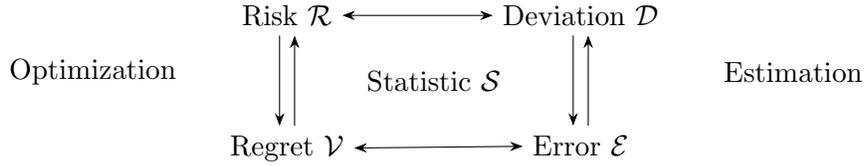
\begin{figure}[ht]
\centering
    \begin{tikzpicture}[
        node distance=1.2cm and 2cm,
        every node/.style={align=center},
        >=Stealth] 
        \node (risk) {Risk $\mathcal{R}$};
        \node (deviation) [right=of risk] {Deviation $\mathcal{D}$};
        \node (regret) [below=of risk] {Regret $\mathcal{V}$};
        \node (error) [below=of deviation] {Error $\mathcal{E}$};
        \node at ($(risk)!0.5!(error)$) {Statistic $\mathcal{S}$};    
        \draw[<->] (risk) -- (deviation);
        \draw[<->] (regret) -- (error);
        \draw[->] ($(risk.south) + (-0.1,0)$) -- ($(regret.north) + (-0.1,0)$);
        \draw[->] ($(regret.north) + (0.1,0)$) -- ($(risk.south) + (0.1,0)$);

        \draw[->] ($(deviation.south) + (-0.1,0)$) -- ($(error.north) + (-0.1,0)$);
        \draw[->] ($(error.north) + (0.1,0)$) -- ($(deviation.south) + (0.1,0)$);
        
        \node (opt) [left=0.6cm of risk, yshift=-0.75cm] {Optimization};
        
        \node (est) [right=0.6cm of deviation, yshift=-0.75cm] {Estimation};
    
    \end{tikzpicture}
    \caption{Fundamental Risk Quadrangle in \cite{Rockafellar2013}.} 
    \label{fig:quadrangle} 
\end{figure}

\begin{figure}[!ht]
    \centering
    \includegraphics[width=0.7\linewidth]{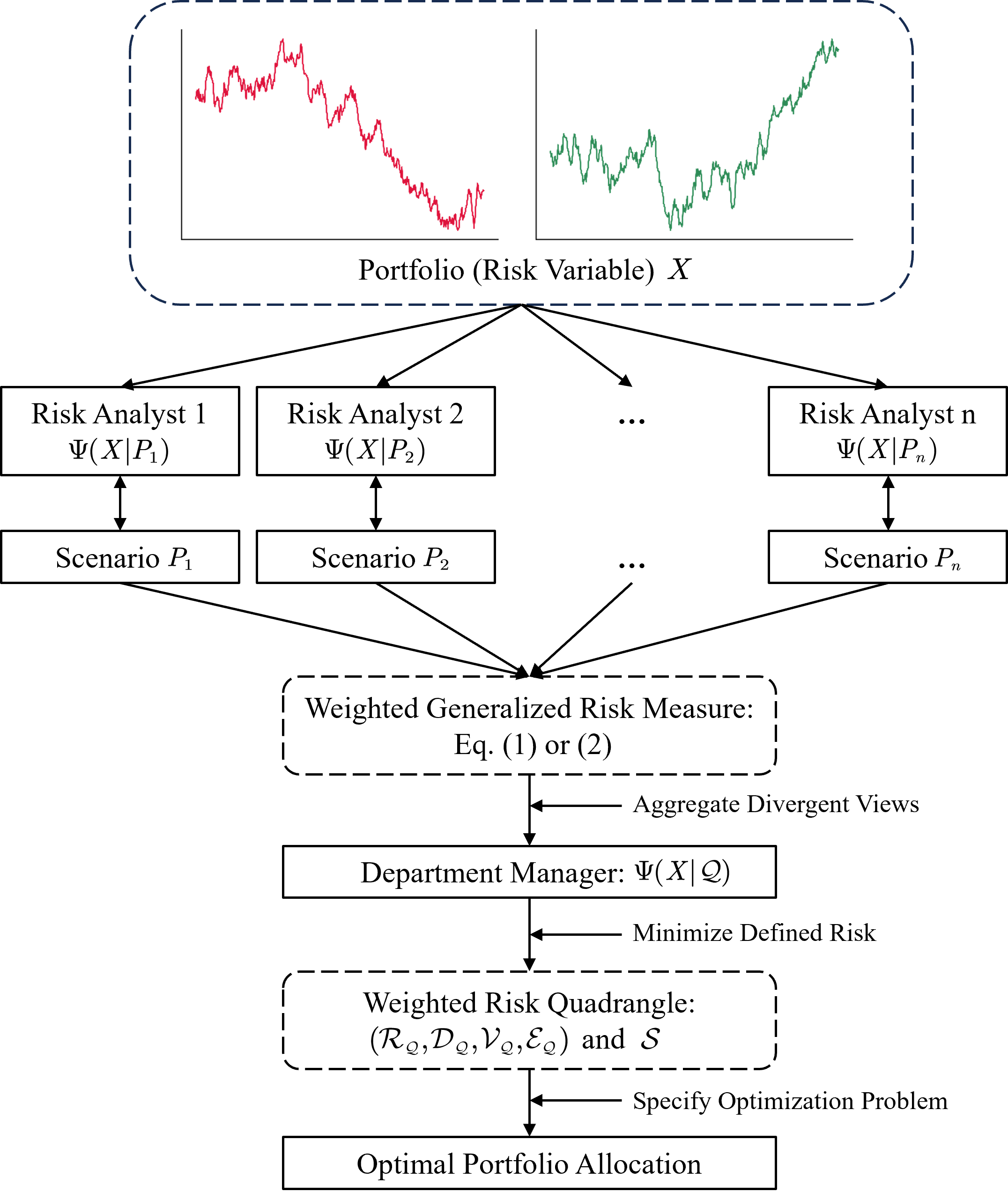}
    \caption{A Schematic Procedure for Department Managers to Aggregate Heterogeneous Risk Assessments.}
    \label{fig:flowchart}
\end{figure}

Before elaborating on the theoretical framework, Figure \ref{fig:flowchart} presents a specific application. Consider a department manager allocating capital across assets based on the assessments from multiple risk analysts. Even when applying a common metric like $\text{ES}_{0.95}$, analysts typically report divergent estimates due to distinct modeling choices and stress scenarios. These differences reflect distinct subjective probability measures $P$ each analyst adopts, collectively forming a family of probability measures $\mathcal{Q}$. Crucially, relying on a single analyst's estimate may expose the portfolio to idiosyncratic model bias, but also not fully make use of others' effort and expertise. 
Thus, the WGRM framework serves as a structural defense to perform a weighted aggregation of the analysts' heterogeneous assessments. The weights assigned to each analyst can be based on their historical prediction accuracy, professional competence, or other relevant criteria. Such process corresponds to assigning a weight $\mu_{\mathcal{Q}}(P)$ to each probability measure $P \in \mathcal{Q}$.
The manager then aims to minimize the aggregated risk while ensuring a predefined level of expected return. Thus, the obtained risk measure by the WGRM can be used as the objective function, which is then integrated into the WRQ. By defining relevant constraints (e.g., expected return targets, position limits) and solving the optimization problem, the manager ultimately obtains a robust asset allocation strategy. By design, this WRQ-based portfolio is inherently resilient, dampening the adverse financial consequences that would arise from relying on an isolated, incorrect market outlook.

\subsection{Connection to Other Frameworks and Contribution}

Our study is directly inspired by and develops the GRM proposed in \cite{Fadina2024}, thus exhibiting substantial differences in both mathematical structure and methodology from traditional risk measures (e.g., \citealt{Frittelli2002, Follmer2002, Föllmer2004, Wakker2010, Castagnoli2022}).
Nevertheless, \cite{Fadina2024} primarily discuss the worst-case GRM, which can be regarded as a special case of the WGRM in this paper under certain regularity conditions. This generality gives rise to disparities in mathematical treatments. Conceptually, our framework requires to accommodate information within each probability measure, whereas the worst-case one relies more on extreme scenarios; 
see also \cite{GS1989,Zhu2009, Zymler2013,Adrian2016,Liu&Wang2021,chen2022ordering,BLLW25}.

While another analogous framework also accounting for multiple scenarios is the scenario-based risk measure in \cite{Wang&Ziegel2021}, they are fundamentally distinct in nature. Specifically, their framework, under the assumption of mutually singular probability measures, emphasizes model uncertainty other than heterogeneous assessments, whereas our framework focuses on the weight structure, particularly on how convex weights are derived. This critical difference leads to distinct methodological treatments, resulting in no overlap in research scopes or analytical approaches.
Other similar weighted aggregation frameworks appear in \cite{KMM2005,Righi2018, Jokhadze2020};  however, these works lack a systematic theoretical characterization of the weight structure itself.

For recent advances in risk measure theory, we refer to \cite{Bellini2017, Wang&Zitikis2021, Fissler2023, Bernard2024, LaTorre2024, Gomez2024, Battiston2025, Cai2025}. Our goal is to provide a comprehensive discussion of the weight structure inherent to the GRM framework; as such, this work may enable further synergistic integration with the aforementioned strands of research, laying groundwork for more comprehensive and context-aware risk management paradigms.

The contributions of this paper include the following aspects. Our primary contribution is that, building upon the foundational work of \cite{Kouetal2013,Kou&Peng2016}, Section \ref{sec:WeightedGRM} establishes analytical characterization of the WGRM framework under both discrete and continuous scenario settings. This dual formulation accommodates the diverse analytical demands of real-world risk management applications. This work represents a crucial deepening of the GRM theory, complementing the worst-case structure that the GRM initially focuses on. Furthermore, we analyze the specific conditions under which the aggregation weights can be uniquely determined. By focusing on the weighting structure itself, the WGRM achieves broad generality, seamlessly adapting to the heterogeneous individual risk measures utilized in the aggregation process.

It must be emphasized that the weighted aggregation of separate risk assessments is distinct from evaluating a single risk measure under a pre-aggregated probability distribution. The latter excessively smooths inherent risk volatility, an information loss that adjusting quantile parameters or other ad hoc modifications cannot adequately offset.

Our second key contribution, detailed in Section \ref{sec:E_quadrangle}, is to introduce the WRQ by extending the FRQ via the proposed WGRM weighting structure. This extension provides an integrated framework that theoretically unifies risk management, optimization, and statistical estimation in multi-scenario settings. While existing literature typically enhances FRQ by aggregating single-measure components with varying parameters, as noted earlier, heterogeneity across distinct scenarios and probability measures remains a more critical dimension to address for practical risk management.
WRQ inherently integrates risk management and portfolio optimization in a cohesive, theoretically consistent manner. Accordingly, in Section \ref{sec:Optimization_WGRM}, we elaborate on the embedded optimization problems within WRQ, as well as how to transform complex risk-minimization problems using WRQ’s components. 

Finally, an empirical validation is provided in Section \ref{sec:empirical}, where we use NASDAQ 100 and S\&P 500 data to compare portfolio performance with and without WGRM (and WRQ) across expansion and recession regimes, and conduct sensitivity analyses. Our results validate that WGRM-based portfolios exhibit superior risk-adjusted performance and pronounced downside resilience, effectively mitigating adverse consequences caused by isolated, erroneous market judgments.
\section{Weighted Generalized Risk Measures}
\label{sec:WeightedGRM}

\subsection{Properties and Technical Discussions}
In this section, we examine the conditions under which the WGRM representation exists. To begin with, we recall several properties of a GRM $\Psi$: $\calX \times 2^{\calP} \to (-\infty, \infty)$ from \cite{Fadina2024} to facilitate a preliminary comparison between the worst-case GRM and WGRM.

The measure is called standard if $\Psi(s|\calQ)=s, \,\forall s \in \R \,\text{and}\, \calQ \subseteq \calP$. When it is a single scenario case, the measure takes the form of $\Psi$: $\calX \times \calP \to (-\infty, \infty)$. Clearly, the explicit inclusion of $\calQ$ as an input generalizes traditional risk measures. When $\calQ$ is restricted to a singleton set (i.e., $\calQ=\{P\}$ for some $P \in \mathcal{P}$), this generalized risk measure reduces to its traditional counterpart. The worst-case GRM satisfies the following properties.

(A1) Scenario monotonicity: if $\Psi(X|P) \leq \Psi(Y|P)$ holds for every $P\in\calQ$, then $\Psi(X|\calQ) \leq \Psi(Y|\calQ)$. 

(A2) Scenario upper bound: $\Psi(X|\calQ) \leq \sup_{P\in \calQ} \Psi(X|P)$ holds for every $X\in \calX$ and $\calQ \subseteq \calP$. 

(A3) Uncertainty aversion: $\Psi(X|\calQ) \leq \Psi(X|\mathcal{R})$ for every $X\in \calX$ and $\calQ \subseteq\mathcal{R} \subseteq \calP$. 

For a WGRM $\Psi$, properties (A1) and (A2) are intuitively justified. Property (A1) holds due to the monotonicity of the integral: $\Psi(X|\calQ)=\int_{\calQ} \Psi(X|P) \,\d\mu_{\calQ}(P) \leq \int_{\calQ} \Psi(Y|P) \,\d\mu_{\calQ}(P) = \Psi(Y|\calQ)$. 
And property (A2) is derived from $\Psi(X|\calQ) \leq \displaystyle\sup_{P^*\in \calQ}\Psi(X|P^*)\int_{\calQ}  \,\d\mu_{\calQ}(P)=\sup_{P^*\in \calQ}\Psi(X|P^*)\cdot \mu_{\calQ}(\calQ)=\sup_{P^*\in \calQ}\Psi(X|P^*)$.

However, property (A3) is rarely satisfied by most of WGRMs but aligns more closely with the worst-case one. The defining characteristic of (A3) is that risk evaluation is monotonically non-decreasing with respect to expanding the scenario set.  This monotonicity may induce excessive conservatism in risk assessment. Intuitively, if posterior scenarios incorporate more favorable states, the aggregate risk measure should exhibit a non-increasing trend rather than a non-decreasing one. This is why the characterization of the WGRM is worth being investigated.

\begin{remark}
	While the codomain of $\Psi$ is formally defined as $(-\infty,  \infty)$, alternative specifications such as $[-\infty, \infty]$ or $(-\infty, \infty]$ are both theoretically admissible. However, from a practical standpoint, risk measures with finite values are generally more interpretable and operationally meaningful in financial applications.
\end{remark}

\begin{remark}
\label{remark:Polish}
    We require that $\mathcal{P}$ is a compact Polish space, which is a simplified technical assumption made for clarity of exposition. The collection $\mathcal{P}$ of all atomless probability measures on $(\Omega,\mathcal{F})$ can be endowed with topologies (e.g., the topology of weak convergence, provided that $(\Omega,\mathcal{F})$ supports a metric making it a Polish space) under which it forms a Polish space. However, $\mathcal{P}$ itself is generally not compact in these natural topologies. For instance, consider $\Omega=\mathbb{R}$ with the Borel $\sigma$-algebra. The sequence of atomless measures $\{P_{n}\}$, where $P_{n}$ is the uniform distribution on the interval $[n,n+1]$, typically lacks a convergent subsequence within $\mathcal{P}$ under weak convergence; the probability mass escapes to infinity. Nevertheless, we can always restrict attention to a suitable compact subset $\widetilde{\mathcal{P}} \subseteq \mathcal{P}$ containing the relevant scenario sets $\mathcal{Q}$ under consideration. Examples include families of atomless measures defined by bounded likelihood ratios with respect to a reference atomless measure (e.g., $\{P\in \mathcal{P}:|\mathrm{d}P / \mathrm{d}P_{0}| \leq M\}$ for some $M < \infty$ and $P_{0}\in \mathcal{P}$) or parametric families of atomless measures (e.g., $\{P_{\theta} \in \Theta\}$ with compact parameter spaces $\Theta$).
\end{remark}

Critically, since $\calQ$ is an arbitrary subset of $\calP$ which may be either discrete or continuous, these two cases exhibit different mathematical structures and must be discussed separately. We first examine the case where $\calQ$ is discrete.

\subsection{WGRM under a Discrete Setting}

In the discrete case where $\mathcal{Q}$ is finite, e.g., of size $n$, our focus is on the Euclidean space $\mathbb{R}^{n}$.  For notational simplicity, we present $\mu_{\mathcal{Q}}$ as a vector, i.e., $\mu_{\mathcal{Q}}=(\mu_1,\dots,\mu_n)^T \in \mathbb{R}^n$. 
Let $\mathcal{C}:=\{\mathbf{x}=(x_1, x_2, \dots,x_n)\in \mathbb{R}^{n} \mid x_{1} \leq x_{2}\leq \dots \leq x_{n} \}$, which is a translation-invariant and closed convex cone. And denote by $\text{int}\mathcal{C}:=\{\mathbf{x}\in \mathbb{R}^{n} \mid x_{1} < x_{2} < \dots < x_{n} \}$ the interior of $\mathcal{C}$.
Let $\mathcal{D}:=\{ \mathbf{x} \in \mathbb{R}^{n} \mid \sum_{i=1}^{n} x_{i}=1, x_{i}\geq 0, i = 1, 2, \dots, n \}$ be the whole weight set in $\mathbb{R}^{n}$. Let $\left \langle \cdot,\cdot \right \rangle$ denote the inner product on $\mathbb{R}^{n}$. For a function $f:\mathbb{R}^{n}\to \mathbb{R}$, we define its domain as $\text{dom}\, f:=\{\mathbf{x}\in \mathbb{R}^{n} \mid f(\mathbf{x}) < \infty \}$ and $f$ is said to be proper if $\text{dom}f \neq \emptyset$. 

When $\calQ$ is finite, we can index its elements as $\calQ = \{P_{i}\}_{i=1}^{n}$.
Generally, the GRM $\Psi(X|\mathcal{Q})$ can be constructed through an aggregation function $f: \mathbb{R}^{n} \to \mathbb{R}$ that combines all individual risk measures under each $P_{i} \in \mathcal{Q}$, i.e.,
\begin{align*}
    \Psi(X|\mathcal{Q}) = f(\Psi(X|P_{1}), \Psi(X|P_{2}),\dots,\Psi(X|P_{n})).
\end{align*}

For simplicity, we define $\Phi_{\calQ, X} = (\Psi(X|P_{1}), \Psi(X|P_{2}),\dots,\Psi(X|P_{n}))^{T} \in \mathbb{R}^{n}$. Thus, the above can be expressed as $\Psi(X|\mathcal{Q}) = f(\Phi_{\calQ, X})$.
To be a satisfactory aggregation function, $f$ is expected to satisfy several properties \citep{Kouetal2013,Kou&Peng2016}:

(B1) Positive homogeneity and translation invariance: 
$f(a \Phi_{\calQ, X}+b \mathbf{1}) = a f(\Phi_{\calQ, X})+ b, \forall\, \Phi_{\calQ, X} \in \mathbb{R}^{n}, a \geq 0, b \in \mathbb{R}$. Here, $\mathbf{1}=(1,1,\dots,1)^{T} \in \mathbb{R}^{n}$. 

(B2) Monotonicity: 
$f(\Phi_{\calQ, X}) \leq f(\Phi_{\calQ, Y})$, if $\Phi_{\calQ, X} \leq \Phi_{\calQ, Y}$, where $\Phi_{\calQ, X} \leq \Phi_{\calQ, Y}$ is understood component-wise, i.e., $\Psi(X|P_{i}) \leq \Psi(Y|P_{i})$, for all $i=1,2,\dots,n$.

(B3) Comonotonic sub-additivity: 
$f(\Phi_{\calQ, X}+\Phi_{\calQ, Y}) \leq f(\Phi_{\calQ, X}) + f(\Phi_{\calQ, Y})$ whenever $\Phi_{\calQ, X}$ and $\Phi_{\calQ, Y}$ are comonotonic, i.e., $(\Psi(X|P_{i})-\Psi(X|P_{j}))(\Psi(Y|P_{i})-\Psi(Y|P_{j})) \geq 0$ for any $i,j \in \{1,2,\cdots,n\}$.

(B4) Permutation invariance: $f(\Phi_{\calQ, X})=f(\Phi_{\calQ, X}^{\pi})$ for every $\pi \in S_{n}$, where $S_{n}$ is the set of all permutations of $\{1, 2, \dots,n\}$ and $\Phi_{\calQ, X}^{\pi}=(\Psi(X|P_{\pi (1)}), \Psi(X|P_{\pi (2)}),\dots,\Psi(X|P_{\pi (n)}))^{T}$ is therefore the corresponding permuted vector.

(B1) indicates that $f$ is robust to affine transformation. 
(B2) is equivalent to (A1) scenario monotonicity. A direct observation is that, under (B1) and (B2), $f$ is Lipschitz continuous with respect to the maximum-norm $\|\cdot\|_{\infty}$, ensuring stability under small perturbations of risk inputs.
(B3) is also a common assumption which focuses on the co-movement consistency across two risk inputs.
(B4), though less frequently stated in literature, is intuitive for risk aggregation, as the overall risk measure should be invariant to the ordering of individual evaluations.

\begin{remark}
    One might argue that (B1) should be $f(a \Phi_{\calQ, X}+b \mathbf{1}) = a f(\Phi_{\calQ, X})+ s\cdot b$, where $s$ is a fixed constant related to each aggregation function $f$. However, since here $f$ is designed for weighted-averaging individual risk measures, it is reasonable to impose the invariance of aggregating constant values, i.e., $s=1$.
\end{remark}

The assumption of (B3) comonotonic sub-additivity holds practical relevance in real-world data contexts. For instance, if each component $\Psi(X|P_{i})$ of $\Phi_{\calQ, X}$ represents an individual risk measure under a specific market scenario, the risk measures of correlated assets across a family of coherent scenarios often exhibit comonotonicity. Nevertheless, the comonotonicity constraint is not sufficiently elegant for a general risk structure. If we instead require the aggregation function $f$ to satisfy full sub-additivity (unrestricted to comonotonic vectors), it will impose more stringent conditions on the resulting weights. Formally, we write

(B3') Full sub-additivity: 
$f(\Phi_{\calQ, X}+\Phi_{\calQ, Y}) \leq f(\Phi_{\calQ, X}) + f(\Phi_{\calQ, Y})$ for any $\Phi_{\calQ, X},\Phi_{\calQ, Y} \in \mathbb{R}^{n}$.


We first propose the following characterization result on worst-case WGRM, where we technically get inspired from \cite{Ahmed2008}.

\begin{theorem}
\label{theo:WeightedGRM_condition_sup}
    (1) The aggregation function  $f: \mathbb{R}^{n} \to \mathbb{R}$ satisfies (B1)-(B4) if and only if there exists a closed convex set of weights $\mathcal{W}_{1} \subseteq \mathcal{D} \subseteq \mathbb{R}^{n}$, such that
    \begin{align}
    \label{eq:weighted_grm_sup}
        f(\Phi_{\calQ, X})
        =
        \sup\limits_{\mu_{\mathcal{Q}} \in \mathcal{W}_{1}} \left\langle \mu_{\mathcal{Q}},\Phi_{\calQ, X}^{q}  \right\rangle,
        \quad
        \forall\, \Phi_{\calQ, X}\in \mathbb{R}^{n},
    \end{align}
    where $\Phi_{\calQ, X}^{q} := \big( \Psi(X|P_{(1)}), \Psi(X|P_{(2)}),\dots,\Psi(X|P_{(n)}) \big)^{T}$ is the vector of order statistics obtained by sorting the individual risk assessments $\{\Psi(X|P_i)\}_{i=1}^n$ into non-decreasing order.
    

    (2) The aggregation function $f: \mathbb{R}^{n} \to \mathbb{R}$ satisfies (B1), (B2), (B3') and (B4)  if and only if there exists a closed convex set of weights $\mathcal{W}_{2} \subseteq \mathcal{D} \cap\mathcal{C} \subseteq \mathbb{R}^{n}$, such that
    \begin{align}
    \label{eq:weighted_grm_sup_subadd}
        f(\Phi_{\calQ, X})
        =
        \sup\limits_{\mu_{\mathcal{Q}} \in \mathcal{W}_{2}} \left\langle \mu_{\mathcal{Q}},\Phi_{\calQ, X}^{q}  \right\rangle,
        \quad
        \forall\, \Phi_{\calQ, X}\in \mathbb{R}^{n}.
    \end{align}
    
\end{theorem}

\begin{remark}
    The key distinction between Parts (1) and (2) in Theorem \ref{theo:WeightedGRM_condition_sup} lies in the origin of convexity. Under full sub-additivity, the aggregation function $f$ inherently possesses convexity, whereas under comonotonic sub-additivity, convexity is induced by restricting the domain to $\mathcal{C}$ due to $\delta(\cdot|\mathcal{C})$. Specifically, considering the setting in Theorem \ref{theo:WeightedGRM_condition_sup}, for any $\Phi_{\calQ, X}^{q} \in \text{int}\mathcal{C}$ with components $\Psi(X|P_{1})<\Psi(X|P_{2})<\cdots <\Psi(X|P_{n})$, the closed and convex nature of $\mathcal{W}_{1}$ guarantees the existence of a weight vector $\mu_{1} \in \mathcal{W}_{1}$ that achieves the supremum of $f(\Phi_{\calQ, X}^{q})$. If we suppose $\mu_{1}$ is not non-decreasing, then there exist indices $1\leq j<k \leq n$ such that $\mu_{1j}>\mu_{1k}$. Define $\mu_{2}$ as the weight vector obtained by swapping the $j$-th and $k$-th components of $\mu_{1}$. Clearly, $\mu_{2}$ still belongs to $\mathcal{D}$. However, direct calculation yields
    $$
    \left\langle \mu_2,\Phi_{\calQ, X}^{q} \right\rangle - \left\langle \mu_1,\Phi_{\calQ, X}^{q} \right\rangle
    =
    (\mu_{1k}-\mu_{1j}) \Psi(X|P_{j}) + (\mu_{1j}-\mu_{1k}) \Psi(X|P_{k})
    =
    (\mu_{1k}-\mu_{1j})(\Psi(X|P_{j}) -\Psi(X|P_{k}))
    >0.
    $$
    This implies $\left\langle \mu_2,\Phi_{\calQ, X}^{q} \right\rangle > \left\langle \mu_1,\Phi_{\calQ, X}^{q} \right\rangle = f(\Phi_{\calQ, X}^{q})$, which contradicts the definition of $\mu_{1}$ as the supremum-achieving weight. Nevertheless, this contradiction does not directly indicate that $\mu_{1} \in \mathcal{C}$, since the swapped vector $\mu_{2}$ may not belong to $\mathcal{W}_{1}$ in the first place. 
\end{remark}

The above theorem provides a formal representation of $\Psi(X|\mathcal{Q})$ via a linear weighting of individual risk assessments $\Psi(X|P_{i})$. However, it should be noted that the weighting vectors $\mu_\mathcal{Q}$ in this representation are not uniquely determined. More precisely, the weights depend functionally on the input risk vector $\Phi_{\calQ, X}$, exhibiting variability as $\Phi_{\calQ, X}$ changes. While this theoretical formulation possesses mathematical elegance, the state-dependent nature of the weights may present practical implementation challenges. Consequently, we shall introduce additional structural constraints to enhance the stability and applicability of the weighting scheme.

(B6) Comonotonic additivity:
$f(\Phi_{\calQ, X}+\Phi_{\calQ, Y}) = f(\Phi_{\calQ, X}) + f(\Phi_{\calQ, Y})$ whenever $\Phi_{\calQ, X}$ and $\Phi_{\calQ, Y}$ are comonotonic, i.e., $(\Psi(X|P_{i})-\Psi(X|P_{j}))(\Psi(Y|P_{i})-\Psi(Y|P_{j})) \geq 0$ for any $i,j \in \{1,2,\cdots,n\}$.

(B6') Full additivity: 
$f(\Phi_{\calQ, X}+\Phi_{\calQ, Y}) = f(\Phi_{\calQ, X}) + f(\Phi_{\calQ, Y})$ for any $\Phi_{\calQ, X},\Phi_{\calQ, Y} \in \mathbb{R}^{n}$.

(B6) transforms the (B3) comonotonic sub-additivity into comonotonic additivity and (B6') transforms the (B3') full sub-additivity into full additivity at the expense of weakening the convexity of the function. The rationality behind this transformation is that the sub-additivity should be represented in $\Psi(X | \cdot)$ rather than in $f(\cdot)$ to avoid overemphasizing the diversification effects. In other words, this aligns with scenarios where hedging effects are primarily captured within the individual risk measures $\Psi(X | \cdot)$ rather than the aggregation process. (B6) and (B6') necessitate the linearity of $f(\Phi_{\calQ, X})$, which enables us to apply the Riesz Representation Theorem to attain unique weighting vectors $\mu_\mathcal{Q}^{*}$.

\begin{proposition}
\label{prop:WeightedGRM_condition_unique}
    (1) The aggregation function $f: \mathcal{C} \to \mathbb{R}$ satisfies (B1), (B2) and (B6) if and only if there exists a unique weighting vector $\mu_\mathcal{Q}^{*} \in \mathcal{D}$, such that
    \begin{align*}
        f(\Phi_{\calQ, X})
        =
        \left\langle \mu_{\mathcal{Q}}^{*},\Phi_{\calQ, X}  \right\rangle,
        \quad
        \forall\, \Phi_{\calQ, X}\in \mathcal{C}.
    \end{align*}

(2) The aggregation function $f: \mathbb{R}^{n} \to \mathbb{R}$ satisfies (B1), (B2), (B4) and (B6) if and only if there exists a unique weighting vector $\mu_\mathcal{Q}^{*} \in \mathcal{D}$, such that
    \begin{align*}
        f(\Phi_{\calQ, X})
        =
        \left\langle \mu_{\mathcal{Q}}^{*},\Phi_{\calQ, X}^{q}  \right\rangle,
        \quad
        \forall\, \Phi_{\calQ, X}\in \mathbb{R}^{n}.
    \end{align*}

(3) The aggregation function $f: \mathbb{R}^{n} \to \mathbb{R}$ satisfies (B1), (B2) and (B6') if and only if there exists a unique weighting vector $\mu_\mathcal{Q}^{*} \in \mathcal{D}$, such that
    \begin{align*}
        f(\Phi_{\calQ, X})
        =
        \left\langle \mu_{\mathcal{Q}}^{*},\Phi_{\calQ, X}  \right\rangle,
        \quad
        \forall\, \Phi_{\calQ, X}\in \mathbb{R}^{n}.
    \end{align*}
    
(4) The aggregation function $f: \mathbb{R}^{n} \to \mathbb{R}$ satisfies (B1), (B2), (B4) and (B6') if and only if
    \begin{align*}
        f(\Phi_{\calQ, X})
        =
        \left\langle \mu_{\mathcal{Q}}^{*},\Phi_{\calQ, X}  \right\rangle,
        \quad
        \forall\, \Phi_{\calQ, X}\in \mathbb{R}^{n},
    \end{align*}
where $\mu_{\mathcal{Q}}^{*} = (\frac{1}{n},\frac{1}{n},\cdots,\frac{1}{n})^{T} \in \mathbb{R}^{n}$.
\end{proposition}


\subsection{WGRM under a Continuous Setting}

The results for the discrete case cover a vast majority of practical applications. However, in situations where the number of distinct individual risk measures requiring aggregation is substantial, or the scenarios under consideration become exceedingly numerous, the elements in $\calQ$ may proliferate toward infinity. In such asymptotic regimes, the existing framework is no longer adequate. Therefore, we are motivated to discuss the WGRM structure in the continuous case. Nevertheless, due to the significant differences between two settings, we begin by establishing the basic framework for the continuous version of WGRM.

In such setting, we assume $\calQ$ to be a closed subset of $\calP$. Let $\mu_{0}$ be a fixed atomless Borel probability reference measure on $\mathcal{Q}$. By the isomorphism theorem for standard probability spaces, there exists a measure-preserving bijection $\tau: (\mathcal{Q},\mu_0) \to ([0,1], \lambda)$, where $\lambda$ denotes the Lebesgue measure on $[0, 1]$. This isomorphism allows us to transfer the analysis from the abstract space $\mathcal{Q}$ to a concrete interval $[0,1]$. Define the \textit{scenario risk functional} $\varphi_X: [0,1] \to \mathbb{R}$ as 
\begin{align}
    \varphi_X(t) = \Psi(X \mid \tau^{-1}(t)), \quad t \in [0,1],
\end{align}
which directly maps points in $[0,1]$ to the corresponding single-scenario risk evaluations. The quantile function (or non-decreasing rearrangement) of $\varphi_X$ is defined as 
\begin{align}
\label{cont:quantile}
    \varphi_X^q(t) = \inf \{s \in \mathbb{R} \mid \lambda({u \in [0,1]: \varphi_X(u) \leq s}) \geq t\}, \quad t \in [0,1].
\end{align}

Previously, in the discrete case, the finite number of scenarios ensures bounded risk vectors. However, in the continuous setting, to ensure the mathematical robustness of the aggregation function, we assume that the scenario risk functional remains bounded.
As a common practice, we consider the space $L^{\infty}([0,1])$.
Consequently, the aggregation function $f$ is defined as a mapping $f: L^{\infty}([0,1]) \to \mathbb{R}$, satisfying $\Psi(X|\mathcal{Q})=f(\varphi_X)$. By restricting the domain to $L^{\infty}([0,1])$, we avoid a pathological behavior such as non-integrability of the scenario risk functional $\varphi_{X}$, thereby guaranteeing the existence and integrability of the quantile function. Note that if $\varphi_X \in L^{\infty}([0,1])$, then $\varphi_X^q \in L^{\infty}([0,1])$ as well, with $\|\varphi_X^q\|_\infty = \|\varphi_X\|_\infty$.

Also, the aggregation function $f$ is expected to satisfy the following properties under the continuous case. These properties are not totally different, but are the counterparts in a continuous context of (B1) to (B4).

(C1) Affine invariance: $f(a \varphi_{X} + b\mathbf{1}_{[0,1]}) = a f(\varphi_{X}) + b$, $\forall\, \varphi_{X} \in L^{\infty}([0,1])$, $a \geq 0$, $b \in \mathbb{R}$, where $\mathbf{1}_{[0,1]}$ is an indicator function over interval $[0,1]$.

(C2) Pointwise monotonicity: $f(\varphi_{X})\leq f(\varphi_{Y})$ if $\varphi_{X} \leq \varphi_{Y}$. Here, $\varphi_{X} \leq \varphi_{Y}$ means $\varphi_X(t) \leq \varphi_Y(t)$ for almost every $t \in [0,1]$.

(C3) Comonotonic sub-additivity: $f(\varphi_{X}+\varphi_{Y}) \leq f(\varphi_{X}) +f(\varphi_{Y})$ whenever $\varphi_{X}, \varphi_{Y}$ are comonotonic, i.e., $(\varphi_X(s) - \varphi_X(t))(\varphi_Y(s) - \varphi_Y(t)) \geq 0$ for almost every $s, t \in [0,1]$.

(C4) Strong permutation invariance: 
$f(\varphi_X) = f(\varphi_X \circ T)$ holds for any measure-preserving transformation $T: [0,1] \to [0,1]$ with respect to the Lebesgue measure $\lambda$, i.e., $\lambda(T^{-1}(A)) = \lambda(A)$ for any Lebesgue measurable set $A \subseteq [0,1]$.

Similar to the discrete case, we can also transform (C3) comonotonic sub-additivity into (C3') full sub-additivity, which yields non-decreasing weighting measures as shown in the following result.

(C3') Full sub-additivity: $f(\varphi_{X}+\varphi_{Y}) \leq f(\varphi_{X}) +f(\varphi_{Y})$, for any $\varphi_{X},\varphi_{Y} \in L^{\infty}([0,1])$.

But unfortunately, relying solely on conditions (C1)-(C4) does not directly yield the infinite-dimensional extension of Theorem \ref{theo:WeightedGRM_condition_sup}. While these parallel properties ensure a reasonable behavior for standard risks, they fail to exclude pathological functionals known as Banach limits (purely finitely additive measures). Consequently, we propose introducing a non-trivial condition to filter out these pathological instances.

(C5) Fatou property: For any sequence $\{\varphi_{X_{n}}\} \subset L^{\infty}([0,1])$ such that $\varphi_{X_{n}} \xrightarrow{a.e.} \varphi_{X} \in L^{\infty}([0,1])$, we have $f(\varphi_{X}) \leq \liminf_{n \to \infty} f(\varphi_{X_{n}})$.

In Theorem \ref{theo:WeightedGRM_condition_sup_continuous} below, condition (C5) is explicitly invoked in Part (1) but omitted in Part (2), since the Fatou property is automatically satisfied under the assumptions of Part (2) (see Theorem 2.2, \cite{Jouini2006}). Based on these considerations, we formally state the following theorem.

\begin{theorem}
\label{theo:WeightedGRM_condition_sup_continuous}
    (1) The aggregation function $f$ on $L^{\infty}([0,1])$ satisfies (C1)-(C5) if and only if there exists a closed convex set of density $\mathcal{W}_{3} \subseteq \{\nu \in L^{1}([0,1]) \mid \int^{1}_{0}\nu(t)\mathrm{d}t=1, \nu(t)\geq 0\ \text{a.e.}\}$ such that
    \begin{align}
        \label{eq:weighted_grm_sup_continuous}
        \Psi(X|\mathcal{Q})=f(\varphi_{X})
        =
        \sup\limits_{\nu \in \mathcal{W}_{3}} \int_{0}^{1} \varphi_{X}^{q}(t) \nu(t)\mathrm{d} t,
        \quad
        \forall\, \varphi_{X} \in L^{\infty}([0,1]).
    \end{align}

    (2) The aggregation function $f$ on $L^{\infty}([0,1])$ satisfies (C1), (C2), (C3') and (C4) if and only if there exists a closed convex set of density $\mathcal{W}_{4} \subseteq \{\nu \in L^{\infty}([0,1]) \mid \int^{1}_{0}\nu(t)\mathrm{d}t=1, \nu(t)\geq 0\ \text{a.e.}, \nu \text{ is non-decreasing}\}$ such that
    \begin{align}
        \label{eq:weighted_grm_sup_continuous_subbadd}
        \Psi(X|\mathcal{Q})=f(\varphi_{X})
        =
        \sup\limits_{\nu \in \mathcal{W}_{4}} \int_{0}^{1} \varphi_{X}^{q}(t) \nu(t)\mathrm{d} t,
        \quad
        \forall\, \varphi_{X} \in L^{\infty}([0,1]).
    \end{align}
\end{theorem}

One may question the expressions in Theorem \ref{theo:WeightedGRM_condition_sup_continuous} differ from our original definition in Eq.(\ref{eq:weighted_grm}). Despite formal differences, they are essentially equivalent. For any density function $\nu \in L^{1}([0,1])$ with $\nu \geq 0$ a.e. and $\int_{0}^{1} \nu(t) \,\mathrm{d}t = 1$, we can define an absolutely continuous probability measure $\mu_{\nu}$ on $[0,1]$ by
\begin{align*}
    \mu_{\nu}(A) = \int_{A} \nu(t) \,\mathrm{d}t,
\end{align*}
for any Lebesgue measurable set $A \subseteq [0,1].$
Conversely, by the Radon-Nikodym Theorem, for any probability measure $\mu_{\mathcal{Q}}$ on $[0,1]$ satisfying $\mu_{\mathcal{Q}} \ll \lambda$, where $\lambda$ denotes the Lebesgue measure, there exists a unique (up to a.e. equivalence) density function $\nu_{\mu} = \frac{\mathrm{d}\mu_{\mathcal{Q}}}{\mathrm{d}\lambda} \in L^{1}([0,1])$ with $\nu_{\mu} \geq 0$ a.e.\ such that
\begin{align*}
    \mu_{\mathcal{Q}}(A) = \int_{A} \nu_{\mu}(t) \,\mathrm{d}t,
\end{align*}
for any Lebesgue measurable set $A \subseteq [0,1].$
This establishes a bijective correspondence (up to a.e. equivalence) between non-negative density functions integrating to $1$ and absolutely continuous probability measures on $[0,1]$.
Now we define two sets of absolutely continuous probability measures:
\begin{align*}
    \widetilde{\mathcal{W}}_3
    &=
    \left\{\mu_{\mathcal{Q}}\in \mathcal{M}_{+}([0,1]) \,\mid\, 
    \mu_{\mathcal{Q}} \ll \lambda,\; 
    \mu_{\mathcal{Q}}([0,1])=1
    \right\},
    \\[6pt]
    \widetilde{\mathcal{W}}_4
    &=
    \left\{\mu_{\mathcal{Q}}\in \mathcal{M}_{+}([0,1]) \,\middle|\, 
    \mu_{\mathcal{Q}} \ll \lambda,\; 
    \mu_{\mathcal{Q}}([0,1])=1,\;
    \frac{\mathrm{d}\mu_{\mathcal{Q}}}{\mathrm{d}\lambda} \text{ is non-decreasing}
    \right\},
\end{align*}
where $\mathcal{M}_{+}([0,1])$ denotes the space of non-negative finite measures on $[0,1]$. By the bijective correspondence established above, we have
\begin{align*}
    \left\{\nu \in L^{1}([0,1]) \,\middle|\, 
    \int_{0}^{1}\nu(t)\,\mathrm{d}t=1,\; 
    \nu(t)\geq 0 \text{ a.e.}
    \right\}
    &=
    \left\{\frac{\mathrm{d}\mu_{\mathcal{Q}}}{\mathrm{d}\lambda} \,\middle|\, 
    \mu_{\mathcal{Q}} \in \widetilde{\mathcal{W}}_3
    \right\},
    \\[6pt]
    \left\{\nu \in L^{1}([0,1]) \,\middle|\, 
    \int_{0}^{1}\nu(t)\,\mathrm{d}t=1,\; 
    \nu(t)\geq 0 \text{ a.e.},\;
    \nu \text{ is non-decreasing}
    \right\}
    &=
    \left\{\frac{\mathrm{d}\mu_{\mathcal{Q}}}{\mathrm{d}\lambda} \,\middle|\, 
    \mu_{\mathcal{Q}} \in \widetilde{\mathcal{W}}_4
    \right\}.
\end{align*}

Then for any $\nu \in \mathcal{W}_3$ or $\mathcal{W}_4$, and for any $g \in L^{\infty}([0,1])$, we have
\begin{align*}
    \int_{0}^{1} g(t) \,\nu(t) \,\mathrm{d}t
    =
    \int_{0}^{1} g(t) \,\mathrm{d}\mu_{\nu}(t).
\end{align*}
This allows us to reformulate Theorem \ref{theo:WeightedGRM_condition_sup_continuous} within the framework of measure-theoretic integration.

\begin{proposition}
\label{prop:measure_representation}
    (1) The aggregation function $f$ on $L^{\infty}([0,1])$ satisfies (C1)-(C5) if and only if there exists a closed convex set of measures $\widetilde{\mathcal{W}}_{3}^{*} :=\{\mu_{\nu} \mid \nu \in \mathcal{W}_3\}\subset \widetilde{\mathcal{W}}_3$ such that 
    \begin{align*}
        \Psi(X|\mathcal{Q}) = f(\varphi_{X})
        =
        \sup_{\mu_{\mathcal{Q}} \in \widetilde{\mathcal{W}}_{3}^{*}} 
        \int_{0}^{1} \varphi_{X}^{q}(t) \,\mathrm{d} \mu_{\mathcal{Q}}(t),
        \quad
        \forall\, \varphi_{X} \in L^{\infty}([0,1]).
    \end{align*}
    
    (2) The aggregation function $f$ on $L^{\infty}([0,1])$ satisfies (C1), (C2), (C3') and (C4) if and only if there exists a closed convex set of measures $\widetilde{\mathcal{W}}_{4}^{*}:=\{\mu_{\nu} \mid \nu \in \mathcal{W}_4\} \subset \widetilde{\mathcal{W}}_4$ such that 
    \begin{align*}
        \Psi(X|\mathcal{Q}) = f(\varphi_{X})
        =
        \sup_{\mu_{\mathcal{Q}} \in \widetilde{\mathcal{W}}_{4}^{*}} 
        \int_{0}^{1} \varphi_{X}^{q}(t) \,\mathrm{d} \mu_{\mathcal{Q}}(t),
        \quad
        \forall\, \varphi_{X} \in L^{\infty}([0,1]).
    \end{align*}
\end{proposition}


Recall that in the discrete setting, replacing (comonotonic) sub-additivity with (comonotonic) additivity guarantees the uniqueness of the weights via the Riesz Representation Theorem. A natural question arises: Can the weighting density $\nu(t)$ in the continuous setting also be unique in this way? Unfortunately, the answer is not straightforward. Unlike the discrete case, working within $L^{\infty}([0,1])$ presents a topological hurdle. Its dual space is not $L^{1}([0,1])$, but rather $(L^{\infty})^{*}$, which includes pathological purely finitely additive measures. Consequently, a standard application of the Riesz Representation Theorem on $L^{\infty}$ does not yield a unique weighting density function.

To secure uniqueness within the $L^{\infty}$ framework, one would typically need to impose additional constraints—such as $L^1$-continuity (see property (D5) below)—to force the functional to behave like an $L^1$ functional, thereby allowing it to be extended to the $L^1$ space. However, rather than navigating this circuitous route of constraining an $L^{\infty}$ functional to mimic the $L^1$ behavior, we find it more mathematically natural to directly shift the underlying domain to $L^{1}([0,1])$. This approach may be of greater interest as $L^{1}([0,1])$ accommodates a broader class of risk profiles.

Define the set $\hat{\mathcal{C}} = \{g:[0,1]\to \mathbb{R}\mid g\ \text{is non-decreasing and left-continuous} \} \subseteq L^{1}([0,1])$. For any $\varphi_{X} \in L^{1}([0,1])$, its quantile function is defined as in Eq.(\ref{cont:quantile}). We examine the following properties under this new setup.

(D1) Affine invariance: $f(a \varphi_{X} + b\mathbf{1}_{[0,1]}) = a f(\varphi_{X}) + b$ for all $\varphi_{X} \in L^{1}([0,1])$, $a \geq 0$, $b \in \mathbb{R}$, where $\mathbf{1}_{[0,1]}$ is an indicator function over interval $[0,1]$.

(D2) Pointwise monotonicity: $f(\varphi_{X}) \leq f(\varphi_{Y})$ if $\varphi_{X} \leq \varphi_{Y}$ almost everywhere on $[0,1]$, i.e., $\varphi_X(t) \leq \varphi_Y(t)$ for almost every $t \in [0,1]$.

(D3) Comonotonic additivity: $f(\varphi_{X}+\varphi_{Y}) = f(\varphi_{X}) + f(\varphi_{Y})$ whenever $\varphi_{X}$ and $\varphi_{Y}$ are comonotonic, i.e., $(\varphi_X(s) - \varphi_X(t))(\varphi_Y(s) - \varphi_Y(t)) \geq 0$ for almost every $s, t \in [0,1]$.

(D3') Full additivity: $f(\varphi_{X}+\varphi_{Y}) = f(\varphi_{X}) + f(\varphi_{Y})$ for all $\varphi_{X}, \varphi_{Y} \in L^{1}([0,1])$.

(D4) Strong permutation invariance: $f(\varphi_X) = f(\varphi_X \circ T)$ for any Lebesgue measure-preserving transformations $T: [0,1] \to [0,1]$, i.e., $\lambda(T^{-1}(A)) = \lambda(A)$ for any Lebesgue measurable set $A \subseteq [0,1]$, where $\lambda$ denotes the Lebesgue measure.

(D5) $L^{1}-$continuity: There exists a constant $M > 0$ such that $|f(\varphi_{X}) - f(\varphi_{Y})| \leq M \|\varphi_{X} - \varphi_{Y}\|_{L^{1}}$ for all $\varphi_{X}, \varphi_{Y} \in L^{1}([0,1])$.

Based on these properties, we present the following results.

\begin{proposition}
\label{prop:WeightedGRM_condition_unique_continuous}
    (1) The aggregation function $f$ satisfies (D1), (D2), (D3), and (D5) if and only if there exists a unique density function $\nu^{*} \in \mathcal{W}_{5} \subseteq \{\nu \in L^{\infty}([0,1]) \mid \int_{0}^{1}\nu(t)\mathrm{d}t=1, \nu(t)\geq 0\ \text{a.e.}\}$ such that
    \begin{align*}
        \Psi(X|\mathcal{Q}) = f(\varphi_{X}) = \int_{0}^{1} \varphi_{X}(t) \nu^{*}(t)\mathrm{d} t, \quad \forall\, \varphi_{X} \in \hat{\mathcal{C}}.
    \end{align*}

    (2) The aggregation function $f$ satisfies (D1), (D2), (D3), (D4), and (D5) if and only if there exists a unique density function $\nu^{*} \in \mathcal{W}_{5} \subseteq \{\nu \in L^{\infty}([0,1])\mid \int_{0}^{1}\nu(t)\mathrm{d}t=1, \nu(t)\geq 0\ \text{a.e.}\}$ such that
    \begin{align*}
        \Psi(X|\mathcal{Q}) = f(\varphi_{X}) = \int_{0}^{1} \varphi_{X}^{q}(t) \nu^{*}(t)\mathrm{d} t, \quad \forall\, \varphi_{X} \in L^{1}([0,1]).
    \end{align*}
    
    (3) The aggregation function $f$ satisfies (D1), (D2), (D3'), and (D5) if and only if there exists a unique density function $\nu^{*} \in \mathcal{W}_{5} \subseteq \{\nu \in L^{\infty}([0,1])\mid \int_{0}^{1}\nu(t)\mathrm{d}t=1, \nu(t)\geq 0\ \text{a.e.}\}$ such that
    \begin{align*}
        \Psi(X|\mathcal{Q}) = f(\varphi_{X}) = \int_{0}^{1} \varphi_{X}(t) \nu^{*}(t)\mathrm{d} t, \quad \forall\, \varphi_{X} \in L^{1}([0,1]).
    \end{align*}

    (4) The aggregation function $f$ satisfies (D1), (D2), (D3'), (D4), and (D5) if and only if
    \begin{align*}
        \Psi(X|\mathcal{Q}) = f(\varphi_{X}) = \int_{0}^{1} \varphi_{X}(t)\mathrm{d} t, \quad \forall\, \varphi_{X} \in L^{1}([0,1]),
    \end{align*}
which implies $\nu^{*}(t) \equiv 1$ almost everywhere on $[0,1]$.
\end{proposition}

\section{Weighted Risk Quadrangle}
\label{sec:E_quadrangle}

WGRM offers a wide range of potential applications. In this section, however, we concentrate specifically on how it extends the Fundamental Risk Quadrangle framework introduced in \cite{Rockafellar2013}. We begin with a detailed explanation of this foundational quadrangle. Several examples are then provided to illustrate the form and implications of WGRM within the expanded quadrangle framework. Although slight notational or presentational variations may occur across different settings, for consistency we adopt the expression given in Eq.(\ref{eq:weighted_grm}) throughout this section, regardless of the discreteness or continuity of the scenario or the uniqueness of the weighting measure $\mu_{\mathcal{Q}}$.

\subsection{Fundamental Risk Quadrangle Overview}

Recall that an overview of the FRQ is illustrated in Figure \ref{fig:quadrangle}. Starting from the upper-left corner of the quadrangle, $\mathcal{R}$ denotes a measure of risk, which quantifies the overall risk inherent in a random variable $X$ (representing loss or cost) by assigning it a numerical value. Such measures are frequently employed in the constraints of optimization problems or used as the objective function. For a constant $c \in \mathbb{R}$, the implicit constraint implied by the uncertain condition $X \leq c$ can be replaced by the explicit inequality $\mathcal{R}(X) \leq c$. This substitution leverages a concrete numerical metric $\mathcal{R}(X)$ to offset the uncertain outcomes, thereby eliminating the uncertainty associated with $X$ and yielding a well-defined constraint.  

Formally, we define $\mathcal{R}$ as a regular measure of risk if $\mathcal{R}$ is a closed convex functional with values in $(-\infty, \infty]$; for any constant $c \in \mathbb{R}$, we have $\mathcal{R}(c) = c$, i.e., the risk of a deterministic cost $c$ equals $c$ itself; for any non-constant random variable $X$, we have $\mathcal{R}(X) > \mathbb{E}[X]$, i.e., the risk measure exhibits risk aversion by assigning a value greater than the expected loss of a non-deterministic $X$.  

In the upper-right corner of the quadrangle, $\mathcal{D}$ denotes a measure of deviation, which quantifies the nonconstancy, i.e., uncertainty inherent in the random variable $X$. Classic examples include the variance and standard deviation. Formally, a regular measure of deviation is defined as a closed convex functional taking values in $[0, \infty]$ such that for any constant $c \in \mathbb{R}$, we have $\mathcal{D}(c) = 0$, while for any non-constant random variable $X$, it holds that $\mathcal{D}(X) > 0$. Notably, symmetry is not imposed as a general requirement. In other words, deviation measures where $\mathcal{D}(X) \neq \mathcal{D}(-X)$ may be of greater interest in practical applications, since downside risk is somewhat always more annoying than upside risk.

In the lower-left corner of the quadrangle, $\mathcal{V}$ denotes a measure of regret, which quantifies the dissatisfaction associated with the potential outcomes—positive, zero, or negative.
Formally, a regular measure of regret is defined as a closed convex functional with values in $(-\infty, \infty]$ that satisfies $\mathcal{V}(0) = 0$; and for any non-zero random variable $X$, it holds that $\mathcal{V}(X) > \mathbb{E}[X]$.

Regret measures are naturally connected to utility measures $\mathcal{U}$, another central concept in decision-making under uncertainty. Specifically, if $X$ is viewed as a loss, then $-X$ corresponds to a gain, and the two concepts are interconvertible through the relation $\mathcal{V}(X) = -\mathcal{U}(-X)$. Equivalently, if $Y$ denotes a gain, the inverse conversion holds: $\mathcal{U}(Y) = -\mathcal{V}(-Y)$.  

In the lower-right corner of the quadrangle, $\mathcal{E}$ denotes a measure of error, which quantifies the nonzeroness of a random variable. Formally, a regular measure of error is defined as a closed convex functional with values in $[0, \infty]$ that satisfies $\mathcal{E}(0) = 0$, and $\mathcal{E}(X) > 0$ for any non-zero random variable $X$; moreover, for any sequence of random variables $\{X_k\}_{k=1}^{\infty}$, if $\lim\limits_{k \to \infty} \mathcal{E}(X_k) = 0$, then we have $\lim\limits_{k \to \infty} \mathbb{E}[X_k] = 0$.

In estimation problems, particularly in regression where a loss random variable $Y$ is approximated by a function $f(X_1, \dots, X_N)$ of other random variables, $\mathcal{E}$ serves as a metric to evaluate the magnitude of the prediction error $Z_f = Y - f(X_1, \dots, X_N)$, thereby assessing the quality of the estimation. As with $\mathcal{D}$, asymmetric measures of error, where $\mathcal{E}(X) \neq \mathcal{E}(-X)$, often warrant greater attention.
\begin{align}
    \label{relation1}
    \mathcal{R}(X) &= \mathbb{E}[X] + \mathcal{D}(X), \quad
    \mathcal{D}(X) = \mathcal{R}(X) - \mathbb{E}[X]; \\
    \label{relation2}
    \mathcal{V}(X) &= \mathbb{E}[X] + \mathcal{E}(X), \quad
    \mathcal{E}(X) = \mathcal{V}(X) - \mathbb{E}[X]; \\
    \label{relation3}
    \mathcal{R}(X) &= \min_{c} \{c + \mathcal{V}(X-c) \}, \quad
    \mathcal{D}(X) = \min_{c} \{ \mathcal{E}(X-c) \}; \\
    \label{relation4}
    \arg\min_{c}& \{ c + \mathcal{V}(X - c) \} = \mathcal{S}(X) = \arg\min_{c} \{ \mathcal{E}(X - c) \}.
\end{align}

The four corner elements of the quadrangle are not independent; rather, they interact and are interconnected through Eqs.(\ref{relation1})–(\ref{relation4}). Specifically, it is noted in Figure \ref{fig:quadrangle} that bidirectional arrows connect the risk measure $\mathcal{R}$ to the deviation measure $\mathcal{D}$, and the regret measure $\mathcal{V}$ to the error measure $\mathcal{E}$—their conversion relations are given by Eqs.(\ref{relation1}) and (\ref{relation2}), respectively. In contrast, there is no direct equivalent conversion relation between $\mathcal{R}$ and $\mathcal{V}$, nor between $\mathcal{D}$ and $\mathcal{E}$; instead, these pairs are linked through $\mathcal{S}$, the statistic located at the center of the quadrangle.  

Concretely, for a given error measure $\mathcal{E}$ and random variable $X$, we seek a constant $c \in \mathbb{R}$ that minimizes $\mathcal{E}(X - c)$. The minimum value attained by this error measure is exactly the deviation measure $\mathcal{D}(X)$ of $X$, and the constant $c$ that achieves this minimum is precisely the statistic $\mathcal{S}(X)$ of $X$, with $\mathcal{S}(X+d) = \mathcal{S}(X) + d$ for any $d \in \mathbb{R}$. This relationship is reflected in the right-hand segments of Eqs.(\ref{relation3}) and (\ref{relation4}). The relationship between $\mathcal{R}$ and $\mathcal{V}$ follows the same logic, as illustrated in the left-hand segments of Eqs.(\ref{relation3}) and (\ref{relation4}).  

For expositional brevity, we refer to the five quantities (i.e., $\mathcal{R}$, $\mathcal{D}$, $\mathcal{V}$, $\mathcal{E}$, and $\mathcal{S}$) in Figure \ref{fig:quadrangle} that satisfy Eqs.(\ref{relation1})–(\ref{relation4}) a quadrangle quartet $(\mathcal{R}, \mathcal{D}, \mathcal{V}, \mathcal{E})$ associated with the statistic $\mathcal{S}$. If the four measures (i.e., $\mathcal{R}$, $\mathcal{D}$, $\mathcal{V}$, and $\mathcal{E}$) within this quartet are all regular, we term this structure a regular quadrangle quartet. In what follows, our primary focus will be on regular quadrangle quartets.

\subsection{Weighted Risk Quadrangle}

While the quadrangle framework ingeniously integrates the five quantities into a unified structure, it suffers from a notable limitation that it is confined to various measures defined under a single probability measure. With the advancement of risk management technologies and the growing complexity of real-world stochastic environments, this single-measure setting has become  restrictive. Thus, in this subsection, we build on the WGRM proposed earlier to further expand this quadrangle framework, which is denoted as Weighted Risk Quadrangle (WRQ), enabling it to accommodate multi-measure and multi-scenario risk modeling.

To avoid notational confusion, we denote the measures of the original FRQ under a single probability measure as $\mathcal{R}_{P}, \mathcal{D}_{P}, \mathcal{V}_{P}, \mathcal{E}_{P}, \mathcal{S}_{P}$, while using $\mathcal{R}_{\mathcal{Q}}, \mathcal{D}_{\mathcal{Q}}, \mathcal{V}_{\mathcal{Q}}, \mathcal{E}_{\mathcal{Q}}, \mathcal{S}_{\mathcal{Q}}$ to represent the corresponding measures within WRQ under a family of probability measures. Leveraging the results from Section \ref{sec:WeightedGRM}, we naturally adopt the proposed WGRM as the risk measure in the WRQ. Specifically, $\mathcal{R}_{\mathcal{Q}}(X) = \Psi(X|\mathcal{Q}) = \int_{\mathcal{Q}} \Psi(X|P) \,\mathrm{d}\mu_{\mathcal{Q}}(P)$. Under the weight $\mu_{\mathcal{Q}}$ defined, the risk quadrangle can be extended to a multi-scenario setting.

Prior to this extension, we first observe that the relationships in Eqs.(\ref{relation1})–(\ref{relation4}) also need to be extended to a multi-scenario context. In the original framework, the expectation operator $\mathbb{E}[X]$ plays a crucial bridging role. However, in the multi-scenario setting, this operator should also be extended to a multi-scenario expectation, defined as $\mathbb{E}_\mathcal{Q}[X] = \int_\mathcal{Q} \mathbb{E}_{P}[X] \mathrm{d} \mu_\mathcal{Q}(P)$. Then we can carry out a parallel generalization of Eqs.(\ref{relation1})–(\ref{relation4}) in the multi-scenario context as follows.
\begin{align}
    \label{Erelation1}
    \mathcal{R}_{\mathcal{Q}} (X) &= \mathbb{E}_{\mathcal{Q}} [X] + \mathcal{D}_{\mathcal{Q}} (X), \quad
    \mathcal{D}_{\mathcal{Q}} (X) = \mathcal{R}_{\mathcal{Q}} (X) - \mathbb{E}_{\mathcal{Q}} [X]; \\
    \label{Erelation2}
    \mathcal{V}_{\mathcal{Q}} (X) &= \mathbb{E}_{\mathcal{Q}}[X] + \mathcal{E}_{\mathcal{Q}} (X), \quad
    \mathcal{E}_{\mathcal{Q}} (X) = \mathcal{V}_{\mathcal{Q}} (X) - \mathbb{E}_{\mathcal{Q}} [X]; \\
    \label{Erelation3}
    \mathcal{R}_{\mathcal{Q}} (X) &= \min_{c} \{c + \mathcal{V}_{\mathcal{Q}} (X-c) \}, \quad
    \mathcal{D}_{\mathcal{Q}} (X) = \min_{c} \{ \mathcal{E}_{\mathcal{Q}} (X-c) \}; \\
    \label{Erelation4}
    \arg\min_{c}& \{ c + \mathcal{V}_{\mathcal{Q}} (X - c) \} = \mathcal{S}_{\mathcal{Q}} (X) = \arg\min_{c} \{ \mathcal{E}_{\mathcal{Q}} (X - c) \}.
\end{align}

The expansion implied by Eqs.(\ref{Erelation1})–(\ref{Erelation4}) requires theoretical justifications. Inspired by the Mixing Theorem by \cite{Rockafellar2013}, we further propose the following theorem.

\begin{theorem}
\label{theo:multi_quadrangle}
    Under the weighting measure $\mu_{\mathcal{Q}}$ defined within $\mathcal{R}_{\mathcal{Q}}(X) = \Psi(X|\mathcal{Q}) = \int_{\mathcal{Q}} \Psi(X|P) \,\mathrm{d}\mu_{\mathcal{Q}}(P)$, let $(\mathcal{R}_{P} = \Psi(\cdot|P), \mathcal{D}_{P}, \mathcal{V}_{P}, \mathcal{E}_{P})$ denote a regular quadrangle quartet with statistic $\mathcal{S}_{P}$ for each $P \in \mathcal{Q}$. A muti-scenario quadrangle quartet  $(\mathcal{R}_{\mathcal{Q}}, \mathcal{D}_{\mathcal{Q}}, \mathcal{V}_{\mathcal{Q}}, \mathcal{E}_{\mathcal{Q}})$ with statistic $S_{\mathcal{Q}}$ is generated by  
    \begin{equation}
    \label{eq:multi_quadrangle}
        \begin{aligned}
            \mathcal{S}_{\mathcal{Q}} (X) &= \int_\mathcal{Q} \mathcal{S}_{P}(X) \mathrm{d} \mu_\mathcal{Q} (P), \\
            \mathcal{R}_{\mathcal{Q}} (X) &= \int_\mathcal{Q} \mathcal{R}_{P}(X) \mathrm{d} \mu_\mathcal{Q} (P), \\
            \mathcal{D}_{\mathcal{Q}} (X) &= \int_\mathcal{Q} \mathcal{D}_{P}(X) \mathrm{d} \mu_\mathcal{Q} (P), \\
            \mathcal{V}_{\mathcal{Q}} (X) &= \min_{b(P), P \in \mathcal{Q}} \left\{\int_\mathcal{Q} \mathcal{V}_{P}(X-b(P)) \mathrm{d} \mu_\mathcal{Q} (P) \ \middle| \ \int_\mathcal{Q} b(P) \mathrm{d} \mu_\mathcal{Q}(P) = 0 \right\}, \\
            \mathcal{E}_{\mathcal{Q}} (X) &= \min_{b(P), P \in \mathcal{Q}} \left\{\int_\mathcal{Q} \mathcal{E}_{P}(X-b(P)) \mathrm{d} \mu_\mathcal{Q} (P) \ \middle| \ \int_\mathcal{Q} b(P) \mathrm{d} \mu_\mathcal{Q}(P) = 0 \right\},
        \end{aligned}
    \end{equation}
    where $b(\cdot)$ is a functional on $\mathcal{Q}$.
\end{theorem}

So far, we have presented the theorem formulation of the WRQ. Next, we provide several relevant examples. Notably, due to the equivalent conversion relationships in Eqs.(\ref{Erelation1}) and (\ref{Erelation2}), as long as we know one of $\mathcal{R}_{\mathcal{Q}}$ and $\mathcal{D}_{\mathcal{Q}}$, and  one of $\mathcal{V}_{\mathcal{Q}}$ and $\mathcal{E}_{\mathcal{Q}}$, we can derive the entire quartet $(\mathcal{R}_{\mathcal{Q}}, \mathcal{D}_{\mathcal{Q}}, \mathcal{V}_{\mathcal{Q}}, \mathcal{E}_{\mathcal{Q}})$. We can then further obtain $\mathcal{S}$ via Eq.(\ref{Erelation3}), thereby constructing the complete quadrangle.  

\begin{example}
\label{ex:quadragle}
    Value-at-Risk (VaR, or quantile) and Expected Shortfall (ES, or Conditional VaR/CVaR) are both common tail-based risk metrics that have gained substantial significance in modern risk management frameworks \citep{BCBS2016, McNeil2015}. However, the latter is often preferred over the former because it accounts for tail risk and satisfies coherency. Thus, in this example, for each $P \in \mathcal{Q}$, we set $\mathcal{R}_{P}(X) = \text{ES}_{\alpha}^{P}(X)$, where $\text{ES}_{\alpha}^{P}(X)$ denotes the $\alpha$-level Expected Shortfall of $X$ under probability measure $P$. For the error measure, we adopt the Koenker-Bassett error with appropriate adjustments to ensure that it projects to the desired $\mathcal{D}_{P}(X)$. Specifically, we define $\mathcal{E}_{P}(X) = \mathbb{E}_{P}\left[\frac{\alpha}{1-\alpha} X_{+} + X_{-}\right]$, where $X_{+} = \max \{0, X\}$ and $X_{-} = \max\{0, -X\}$. Using the conversion relations in Eqs.(\ref{Erelation1})–(\ref{Erelation4}), we can derive the remaining quantities of the single-measure quadrangle: $\mathcal{D}_{P}(X) = \text{ES}_{\alpha}^{P}(X - \mathbb{E}_{P}[X])$, $\mathcal{V}_{P}(X) = \frac{1}{1-\alpha} \mathbb{E}_{P}[X_{+}]$, and $\mathcal{S}_{P}(X) = \text{VaR}_{\alpha}^{P}(X)$. Leveraging Theorem \ref{theo:multi_quadrangle}, we immediately obtain the corresponding WRQ:
    \begin{equation*}
	\begin{aligned}
		\mathcal{S}_{\mathcal{Q}}(X) &= \int_\mathcal{Q} \text{VaR}_{\alpha}^{P}(X) \mathrm{d} \mu_\mathcal{Q} (P), \\
		\mathcal{R}_{\mathcal{Q}} (X) &= \int_\mathcal{Q} \text{ES}_{\alpha}^{P}(X) \mathrm{d} \mu_\mathcal{Q} (P), \\
        \mathcal{D}_{\mathcal{Q}} (X) &= \int_\mathcal{Q} \text{ES}_{\alpha}^{P}(X - \mathbb{E}_\mathcal{Q}[X]) \mathrm{d} \mu_\mathcal{Q} (P), \\
        \mathcal{V}_{\mathcal{Q}} (X) &= \min_{b(P), P \in \mathcal{Q}} \left\{\int_\mathcal{Q} \frac{1}{1-\alpha} \mathbb{E}_{P}[(X-b(P))_{+}] \mathrm{d} \mu_\mathcal{Q} (P) \ \middle|\ \int_\mathcal{Q} b(P) \mathrm{d} \mu_\mathcal{Q}(P) = 0 \right\}, \\
        \mathcal{E}_{\mathcal{Q}} (X) &= \min_{b(P), P \in \mathcal{Q}} \left\{\int_\mathcal{Q} \mathbb{E}_{P}\left[\frac{\alpha}{1-\alpha} (X-b(P))_{+} + (X - b(P))_{-}  \right] \mathrm{d} \mu_\mathcal{Q} (P) \ \middle| \ \int_\mathcal{Q} b(P) \mathrm{d} \mu_\mathcal{Q}(P) = 0 \right\}.
    \end{aligned}
    \end{equation*}
    
    For appropriately chosen values of the $\alpha$-level, each individual quadrangle quartet $(\mathcal{R}_{P}, \mathcal{D}_{P}, \mathcal{V}_{P}, \mathcal{E}_{P})$ corresponding to a single probability measure $P \in \mathcal{Q}$ is regular. Thus, the resulting multi-scenario quadrangle quartet $(\mathcal{R}_{\mathcal{Q}}, \mathcal{D}_{\mathcal{Q}}, \mathcal{V}_{\mathcal{Q}}, \mathcal{E}_{\mathcal{Q}})$ is also regular. 
\end{example}
\begin{remark}
    Notably, in Example \ref{ex:quadragle}, neither $\mathcal{D}_{\mathcal{Q}}(X)$ nor $\mathcal{E}_{\mathcal{Q}}(X)$ is necessarily symmetric. This asymmetry inherits from the single-measure components. The $\text{ES}_{\alpha}^{P}$ operator in $\mathcal{D}_{P}$ preserves the asymmetry of asymmetric distributions. For any random variable $X$ , it holds that $\text{ES}_{\alpha}^{P}(-X) = \mathbb{E}_{P}\left[-X \mid -X \geq \text{VaR}_{\alpha}^{P}(-X)\right] = -\mathbb{E}_{P}\left[X \mid X \leq \text{VaR}_{1-\alpha}^{P}(X)\right]$, which shows that $\text{ES}_{\alpha}^{P}(-X) \neq -\text{ES}_{\alpha}^{P}(X)$ in general, thus directly introducing asymmetry into $\mathcal{D}_{P}$. Additionally, the single-measure error measure $\mathcal{E}_{P}(X)$ assigns distinct weights to the positive and negative parts of $X$; specifically, a weight of $\frac{\alpha}{1-\alpha}$ to $X_{+}$ and a weight of $1$ to $X_{-}$. This weighted distinction between $X_{+}$ and $X_{-}$ inherently makes $\mathcal{E}_{P}(X)$ asymmetric. Crucially, this asymmetry carries over to the weighted aggregation.
\end{remark}
\section{Optimization under WGRM and WRQ}
\label{sec:Optimization_WGRM}
This section presents a discussion on the connection between the WGRM, the WRQ, and optimization problems. In a typical optimization setting, the random variable $X$, representing loss, generally depends on a decision vector $\mathbf{x} \in \mathcal{G} \subseteq \mathbb{R}^{m}$, where $\mathcal{G}$ denotes the feasible region. The objective is to determine an optimal decision vector $\mathbf{x} = (x_{1}, \dots,x_{m})^T$ within $\mathcal{G}$ that optimizes an objective function involving $X(\mathbf{x})$.
While the above theoretical framework accommodates both discrete and continuous scenarios, discrete scenarios are predominant in most practical optimization applications. Consequently, Eq.(\ref{eq:multi_quadrangle}) can be reformulated for a discrete set $\mathcal{Q} = \{P_{i} \in \mathcal{P} \mid i=1,\cdots, n \}$ as follows.
\begin{equation*}
\begin{aligned}
	\mathcal{S}_{\mathcal{Q}}(X(\mathbf{x})) &= \sum_{i = 1}^{n} \mathcal{S}_{P_{i}}(X(\mathbf{x})) \cdot \mu_{\mathcal{Q}}(P_{i}), \\
    \mathcal{R}_{\mathcal{Q}}(X(\mathbf{x})) &= \sum_{i = 1}^{n} \mathcal{R}_{P_{i}}(X(\mathbf{x})) \cdot \mu_{\mathcal{Q}}(P_{i}), \\
    \mathcal{D}_{\mathcal{Q}}(X(\mathbf{x})) &= \sum_{i = 1}^{n} \mathcal{D}_{P_{i}}(X(\mathbf{x})) \cdot \mu_{\mathcal{Q}}(P_{i}), \\
    \mathcal{V}_{\mathcal{Q}}(X(\mathbf{x})) &= \min_{\{b(P_{i})\}_{i=1}^{n}} \left\{\sum_{i = 1}^{n} \mathcal{V}_{P_{i}}(X(\mathbf{x})-b(P_{i})) \cdot \mu_{\mathcal{Q}}(P_{i}) \ \middle| \ \sum_{i = 1}^{n} b(P_{i}) \cdot \mu_{\mathcal{Q}}(P_{i}) =0  \right\}, \\
    \mathcal{E}_{\mathcal{Q}}(X(\mathbf{x})) &= \min_{\{b(P_{i})\}_{i=1}^{n}} \left\{\sum_{i = 1}^{n} \mathcal{E}_{P_{i}}(X(\mathbf{x})-b(P_{i})) \cdot \mu_{\mathcal{Q}}(P_{i}) \ \middle | \ \sum_{i = 1}^{n} b(P_{i}) \cdot \mu_{\mathcal{Q}}(P_{i}) =0  \right\}.
\end{aligned}
\end{equation*}

A straightforward observation is that the expressions for $\mathcal{V}_{\mathcal{Q}}$ and $\mathcal{E}_{\mathcal{Q}}$ are inherently minimization problems. 
While the proof of Theorem \ref{theo:multi_quadrangle} directly indicates that $b(P_{i}) = \mathcal{S}_{P_{i}}(X(\mathbf{x})) - \sum_{j = 1}^{n} \mathcal{S}_{P_{j}}(X(\mathbf{x})) \cdot \mu_{\mathcal{Q}}(P_{j})$, which in turn allows us to derive an explicit expression for $\mathcal{V}_{\mathcal{Q}}(X(\mathbf{x}))$, practical implementation may favor numerical solutions via linear programming. This is particularly advantageous when confronting issues such as the complexity and nonlinearity of the individual functional $\mathcal{V}_{P_{i}}$, which makes linear programming more computationally tractable. Specifically, under a regular WRQ, each $\mathcal{V}_{P_{i}}$ is a closed convex functional. Then there exists a family of linear functions $\phi_{ik}$, where $k \in K_{i}$ and $K_{i}$ denotes an index set, such that $\mathcal{V}_{P_{i}}(X(\mathbf{x})) = \sup\limits_{k \in K_{i}} \phi_{ik}(X(\mathbf{x}))$. Thus,  when $X(\mathbf{x})$ is linear in $\mathbf{x}$, by introducing auxiliary variables $t_{i}$, the optimization problem in $\mathcal{V}_{\mathcal{Q}}(X(\mathbf{x}))$ can be transformed into the following linear program:
\begin{equation*}
\begin{aligned}
	\min_{\{b(P_{i})\}_{i=1}^{n}}  \quad &\sum_{i=1}^{n}  \mu_{\mathcal{Q}}(P_{i}) \cdot t_{i} \\
    \text{s.t.} \quad &t_{i} \geq  \phi_{ik}(X(\mathbf{x}) - b(P_{i})), \quad \forall\, i=1,\dots,n, \, k \in K_{i}, \\
    & \sum_{i=1}^{n} b(P_{i}) \cdot \mu_{\mathcal{Q}}(P_{i}) = 0, \\
    & b(P_{i}) \in \mathbb{R}, \quad \forall\, i =1,\dots,n.
\end{aligned}
\end{equation*}

Another set of optimization problems inherent to the WRQ is encapsulated in Eq.(\ref{Erelation3}). Taking the relationship between $\mathcal{R}_{\mathcal{Q}}$ and $\mathcal{V}_{\mathcal{Q}}$ as an example, a typical risk management optimization problem might aim to minimize $\mathcal{R}_{\mathcal{Q}}(X(\mathbf{x}))$ with respect to the decision vector $\mathbf{x}$. Notably, leveraging the key relation $\mathcal{R}_{\mathcal{Q}} (X(\mathbf{x})) = \min\limits_{c \in \mathbb{R}} \left\{c + \mathcal{V}_{\mathcal{Q}} (X(\mathbf{x}) - c) \right\}$, the constraint $\mathcal{R}_{\mathcal{Q}}(X(\mathbf{x})) \leq d$ for any $d \in \mathbb{R}$ is equivalent to the existence of some $c_{0} \in \mathbb{R}$ such that $c_{0} + \mathcal{V}_{\mathcal{Q}} (X(\mathbf{x}) - c_{0}) \leq d$.
Consequently, the original objective function $\min\limits_{\mathbf{x}} \mathcal{R}_{\mathcal{Q}}(X(\mathbf{x}))$ is equivalent to $\min\limits_{\mathbf{x}, c} \left\{c + \mathcal{V}_{\mathcal{Q}} (X(\mathbf{x}) - c) \right\}$, which can be further transformed using the linear programming result above. Namely:
\begin{equation}
\label{eq:quadragle_optimization}
  \begin{minipage}{0.2\linewidth}
    $\begin{aligned}
      \min_{\mathbf{x}}\ &\mathcal{R}_{\mathcal{Q}}(X(\mathbf{x})) \\
      \text{s.t.} \, \ & \mathbf{x} \in \mathcal{G}, \\
      &\text{other constraints}.
    \end{aligned}$
  \end{minipage}
 \longleftrightarrow 
  \begin{minipage}{0.5\linewidth}
    $\begin{aligned}
      \min_{\{b(P_{i})\}_{i=1}^{n}, \mathbf{x}, c} & c + \sum_{i=1}^{n} \mu_{\mathcal{Q}}(P_{i}) \cdot t_{i} \\
      \text{s.t.} \quad & t_{i} \geq \phi_{ik}(X(\mathbf{x}) - c - b(P_{i})), \quad \forall\, i=1,\dots,n, \, k \in K_{i}, \\
      & \sum_{i=1}^{n} b(P_{i}) \cdot \mu_{\mathcal{Q}}(P_{i}) = 0, \\
      & b(P_{i}) \in \mathbb{R}, \quad \forall\, i =1,\dots,n, \\
      & \mathbf{x} \in \mathcal{G}, \\
      &\text{other constraints}.
    \end{aligned}$
  \end{minipage}
\end{equation}

\begin{example}
\label{ex:portfolio_optimization}
    We present here an application of Eq.(\ref{eq:quadragle_optimization}) in portfolio management \citep{Zhu2009, Behera2025}. It is assumed that there are $m$ financial assets in the market, with their returns denoted by random variables $R_{j}$ ($j=1,\dots,m$) and corresponding expected values $\theta_{j}$. We can thus use $-R_{j}$ to represent their potential losses. Our goal is to determine the optimal investment weights $\mathbf{x} = (x_{1},\dots,x_{m})^{T}$ for each assest by incorporating the heterogeneous assessments of $n$ analysts on $m$ financial assets. The total loss of the portfolio can be expressed as $X(\mathbf{x}) = \sum\limits_{j=1}^{m} -x_{j} R_{j}$.

    Following the approach in Example \ref{ex:quadragle}, we adopt $\text{ES}_{\alpha}$ as the portfolio risk metric (see \cite{BCBS2016}). Each analyst forms a view on the probability distribution of $X(\mathbf{x})$ based on their own insight, leading to the individual risk measure $\mathcal{R}_{P_{i}}(X(\mathbf{x})) = \text{ES}_{\alpha}^{P_{i}}(X(\mathbf{x}))= \text{ES}_{\alpha}^{P_{i}}\left(\sum\limits_{j=1}^{m} -x_{j} R_{j}\right)$. Due to differences in analysts’ experience, professional competence, historical performance, and other factors, their perspectives are assigned different weights $\mu_{\mathcal{Q}}(P_{i})$ by the department manager. The corresponding optimization problem is therefore to select the optimal asset allocation weights that minimize the portfolio’s ES, given a target expected portfolio return (denoted by $\theta_{0}$). The objective function of this problem can be written as:
    \begin{align}
    \label{objective}
    \min\limits_{\mathbf{x}} \mathcal{R}_{\mathcal{Q}}(X(\mathbf{x})) = \min\limits_{\{x_{j}\}_{j=1}^{m}} \sum\limits_{i=1}^{n} \text{ES}_{\alpha}^{P_{i}}\left(\sum\limits_{j=1}^{m}-x_{j} R_{j}\right) \cdot \mu_{\mathcal{Q}}(P_{i}).
    \end{align}
    
    However, directly solving (\ref{objective}) using the relation $\text{ES}_{\alpha}^{P_{i}}(X) = \frac{1}{1-\alpha} \int_{\alpha}^{1} \text{VaR}_{\beta}^{P_{i}}(X) \mathrm{d} \beta$ requires a complete estimation of the cumulative distribution function, which is generally unavailable. Notably, within the risk quadrangle, the regret measure corresponding to the risk measure $\mathcal{R}_{P_{i}}(X) = \text{ES}_{\alpha}^{P_{i}}(X)$ is $\mathcal{V}_{P_{i}}(X) = \frac{1}{1-\alpha}\mathbb{E}^{P_{i}}[X_{+}]$. This regret measure, for each Analyst $i$, can be estimated using a series of historical data $X_{ik}$ ($k=1,\dots,T_{i}$), specifically as $\frac{1}{1-\alpha} \cdot \frac{1}{T_{i}} \sum\limits_{k=1}^{T_{i}} (X_{ik})_{+}$. Thus, by introducing auxiliary variables $t_{ik}$, this optimization problem can be transformed into the following linear program:
    \begin{equation}
    \label{eq:portfolio_optimization}
    \begin{aligned}
    	\min_{\{x_{j}\},c,\{b(P_{i})\}, \{t_{ik}\}} \quad & c + \frac{1}{1-\alpha} \sum_{i=1}^{n} \mu_{\mathcal{Q}}(P_{i}) \left( \frac{1}{T_{i}} \sum\limits_{k=1}^{T_{i}} t_{ik} \right)
    	\\
    	\text{s.t.} \quad 
    	& t_{ik} \geq \sum_{j=1}^{m} -x_{j} r_{jk} - c - b(P_{i}),\quad \forall\, k=1,\dots,T_{i},\ \forall\, i=1,\dots,n, \\
    	& t_{ik} \geq 0, \quad \forall\, k=1,\dots,T_{i},\quad \forall\, i=1,\dots,n, \\
    	& \sum_{i=1}^{n} b(P_{i}) \cdot \mu_{\mathcal{Q}}(P_{i}) = 0, \\
        & \sum_{j=1}^{m} x_{j} \theta_{j} = \theta_{0}, \\
        & \sum_{j=1}^{m} x_{j} = 1,\ 0 \leq x_{j} < 1, \quad \forall\, j =1,\dots,m.
    \end{aligned}
    \end{equation}
\end{example}
\begin{remark}
    It is worth noting that in Example \ref{ex:portfolio_optimization}, the specific reflection of each analyst’s subjective probability measure $P_i$ lies in the selection of historical data when estimating $\mathbb{E}_{P_{i}}\left[ \left( -\sum_{j=1}^{m} r_{j}x_{j} - c - b(P_{i}) \right)_{+} \right]$. This estimation $\frac{1}{T_{i}} \sum_{k=1}^{T_{i}} \left( \sum_{j=1}^{m} -x_{j}r_{jk} - c - b(P_{i}) \right)_{+}$ relies on empirical average, where the choice of historical data encodes the analysts’ perspectives on future market dynamics.
    Particularly, the length of the time horizon $T_{i}$ reflects Analyst $i$’s view on the differences between historical and future market conditions. A longer $T_{i}$ indicates a belief that historical patterns are more representative of future trends without much structural change, while a shorter $T_{i}$ signals a perception of greater market divergence from the past. Meanwhile, the selection of specific historical return data $r_{jk}$ embodies an analyst’s judgment on potential future market cycles (e.g., expansion, recession, or stability). For instance, if an analyst believes the future will be characterized by high growth and low inflation without huge structural changes to fundamental market conditions, they will select return data from historical periods with similar macroeconomic contexts and adopt a longer time horizon.
\end{remark}


\section{Empirical Study}
\label{sec:empirical}

In this section, we illustrate the methodology introduced in Example \ref{ex:portfolio_optimization} and apply it to the stock market. Using real-world examples, we conduct a discussion on the advantages and limitations of portfolio optimization leveraging the WGRM and WRQ. Specifically, we select constituent stocks of the NASDAQ 100 Index as the underlying assets, compare portfolio returns under different optimization methods with the index return across both expansion and recession market conditions, and subsequently perform sensitivity analysis by varying model parameters. The data used in this study are publicly available and were retrieved from \url{https://finance.yahoo.com}.

\subsection{Empirical Setup}

We assume there are 4 analysts ($n=4$) with heterogeneous assessments in future macroeconomic trends. Based on their outlooks, each analyst estimates stock returns and informs stock selection in the portfolio, resulting in four distinct scenarios (Scenarios 1–4). The department manager integrates these four analysts' assessments to make the optimal investment decision.
Analyst 1 expects future interest rates to rise. Thus, when using historical data to estimate expected stock returns, they select periods where interest rates exceeded the \textit{median} within the observation window, yielding estimated expected returns $\{\theta^{1}_{j}\}_{j=1}^{m}$. Analyst 2 anticipates future interest rates to fall, so they use data from periods where interest rates were below the median to estimate expected returns, denoted by $\{\theta^{2}_{j}\}_{j=1}^{m}$.
Analyst 3 and Analyst 4 focus primarily on inflation, emphasizing the impact of real economic fluctuations on stock prices. Thus, Analyst 3 predicts rising inflation, with corresponding expected returns denoted by $\{\theta^{3}_{j}\}_{j=1}^{m}$, while Analyst 4 predicts falling inflation, with expected returns $\{\theta^{4}_{j}\}_{j=1}^{m}$. For this study, we use the U.S. 10-year Treasury bond yield as the interest rate proxy since its maturity is more aligned with that of stocks, and use the CPI for inflation measurements. For each analyst $i$, the corresponding portfolio optimization problem is formulated as:
\begin{equation}
\begin{aligned}
	\min_{\{x_{j}\}, c, \{t_{ik}\}}  \quad & c + \frac{1}{(1-\alpha)T_{i}} \sum_{k=1}^{T_{i}} t_{ik}\\
	\text{s.t.} \quad 
	& t_{ik} \geq \sum_{j=1}^{m}-x_{j}r_{jk}-c,\quad \forall\,k=1,\dots,T_{i}, \\
	& t_{ik} \geq 0, \quad \forall\,k=1,\dots,T_{i}, \\
    & \sum_{j=1}^{m} x_{j} \theta^{i}_{j} \geq \theta_{0}, \\
    & \sum_{j=1}^{m} x_{j} =1,\ 0 \leq x_{j} < 1, \forall\,\ j =1,\dots,m.
\end{aligned}
\end{equation}

The manager integrates the four analysts’ assessments, so their expected stock returns are a weighted aggregation of the analysts’ estimates, i.e., $\theta_{j} = \sum\limits_{i=1}^{4} \mu_{\mathcal{Q}}(P_{i}) \theta^{i}_{j}$. For simplicity, we assign equal weights to each analyst’s perspective, i.e., $\mu_{\mathcal{Q}}(P_{i}) = \frac{1}{4}$, adopting a purely synthetic lens to evaluate the framework. Thus, the manager’s optimization problem is exactly the one formulated in Eq.(\ref{eq:portfolio_optimization}).

To conduct a comprehensive assessment, we examine two distinct market regimes, recession and expansion respectively, to evaluate how the manager using the integrated multi-scenario approach performs compared to individual analysts relying on single-scenario frameworks. For the ES metric, we adopt the commonly used level $\alpha = 0.95$. During February and March 2025, the NASDAQ 100 Index experienced a sharp decline, with a cumulative two-month return of $-11.86\%$. Although this episode reflects market turbulence rather than a macroeconomic recession, we refer to it as a recession regime for terminological convenience. For this scenario, we use a total time window of $T = 150$ trading days, spanning from August 23, 2024, to March 31, 2025, and split the data into a training set and a backtesting set using February 1, 2025, as the cutoff date. The target return $\theta_0$ is set to the average daily return of January 2025, calculated as $\theta_0 = 1.64\% / 20$. We construct portfolios using weights estimated from the training set, evaluate their performance using backtesting set data, and compare the results to the actual NASDAQ 100 Index return.

In contrast, September and October 2025 saw robust growth in the index, with a cumulative two-month return of $10.58\%$, designated as the expansion regime. Following a consistent methodology, we again use a 150-trading-day time window (May 12, 2025, to October 30, 2025), splitting the data at September 1, 2025, into two sets. The target return $\theta_0$ is the average daily return of August 2025, computed as $\theta_0 = 3.90\% / 21$, while all other parameters remain identical to the recession regime.

We use these two regimes as benchmark cases, and then conduct sensitivity analysis by adjusting key model parameters to test the framework’s robustness. Specifically, within each regime, we modify the time window to $T = 120$ and $T = 180$ trading days, adjust the ES level to $\alpha = 0.90$ and $\alpha = 0.99$, scale the target return $\theta_0$ by a factor of 2 and 0.5, and finally replace the underlying assets from NASDAQ 100 constituents with S\&P 500 constituents.

Intuitively, since ES explicitly accounts for tail risks, ES-based portfolios exhibit inherent robustness. While this conservative focus on downside protection may moderate returns during market expansions, it should deliver superior resilience during recessions to the broader index. Furthermore, the manager's equal-weighted aggregation of conflicting scenarios (i.e., opposing views on interest rates and inflation) achieves a form of \textit{view diversification}. This prevents overexposure to any single bullish forecast, effectively mitigating the idiosyncratic risk of an isolated misjudgment and lowering overall portfolio volatility. Therefore, to comprehensively evaluate the portfolios beyond raw two-month returns, we incorporate annualized Sharpe and Sortino ratios. By explicitly penalizing total and downside volatility, these metrics properly align with our focus on tail-risk mitigation. The one-year U.S. Treasury yield (3.65\%) serves as the risk-free rate.

\subsection{Baseline Results}

Table \ref{table: baseline} presents the baseline results. During the recession regime, all ES-based portfolios demonstrate excellent downside resilience by delivering positive returns despite a substantial index decline. Although Analyst 1 achieves the highest raw return, the manager's portfolio yields the highest annualized Sharpe and Sortino ratios, generating the highest risk-adjusted return per unit of risk. This indicates the robustness of the WGRM-based portfolios during market distress.
In the expansion regime, the inherent conservatism of ES optimization causes all portfolios to underperform the strongly rallying index. Crucially, however, while Analyst 1's isolated forecast results in a negative return, the manager's portfolio maintains moderate profitability and mid-tier risk-adjusted performance.
This outcome directly underscores the method's inherent trade-off: while the WGRM-based approach may lag during aggressive market expansions, it excels in distressed markets by providing robust downside protection and structurally buffering against the idiosyncratic risk of a single erroneous forecast.

\begin{table}[htbp]
    \centering
    \small
    \caption{Baseline Results.}
    \label{table: baseline}
    \begin{threeparttable}
        \begin{adjustbox}{center}
        \begin{tabular}{lcccccc}
        \toprule
         & \multicolumn{3}{c}{Recession regime} & \multicolumn{3}{c}{Expansion regime} \\
        \cmidrule(lr){2-4} \cmidrule(lr){5-7}
        Portfolio & \begin{tabular}{@{}c@{}}Two-month\\return\end{tabular} & \begin{tabular}{@{}c@{}}Sharpe\\ratio\end{tabular} & \begin{tabular}{@{}c@{}}Sortino\\ratio\end{tabular} & \begin{tabular}{@{}c@{}}Two-month\\return\end{tabular} & \begin{tabular}{@{}c@{}}Sharpe\\ratio\end{tabular} & \begin{tabular}{@{}c@{}}Sortino\\ratio\end{tabular} \\
        \midrule
        Analyst 1 & 5.98\% & 2.574 & 4.387 & -2.81\% & -1.893 & -2.610 \\
        Analyst 2 & 1.85\% & 0.713 & 1.355 & 3.93\% & 1.533 & 3.514 \\
        Analyst 3 & 4.27\% & 1.895 & 3.890 & 6.93\% & 3.307 & 5.224 \\
        Analyst 4 & 4.56\% & 1.723 & 2.584 & 1.67\% & 0.432 & 0.671 \\
        Manager & 5.71\% & 2.809 & 5.336 & 3.28\% & 1.665 & 2.857 \\
        \midrule
        Index two-month return & -11.86\% & & & 10.58\% & & \\
        \bottomrule
        \end{tabular}
        \end{adjustbox}
    \end{threeparttable}
\end{table}

Figure \ref{fig:portfolio_daily_return} visualizes the daily portfolio returns across both market regimes. The manager's portfolio exhibits lower volatility than any single analyst's strategy. By avoiding extreme return fluctuations and structurally reducing both the frequency and magnitude of deep drawdowns, the WGRM-based method confirms its distinctly conservative and downside-protective risk profile.

\begin{figure}[!htbp]
    \centering
    \begin{subfigure}[b]{0.9\textwidth} 
        \centering
        \includegraphics[width=\textwidth]{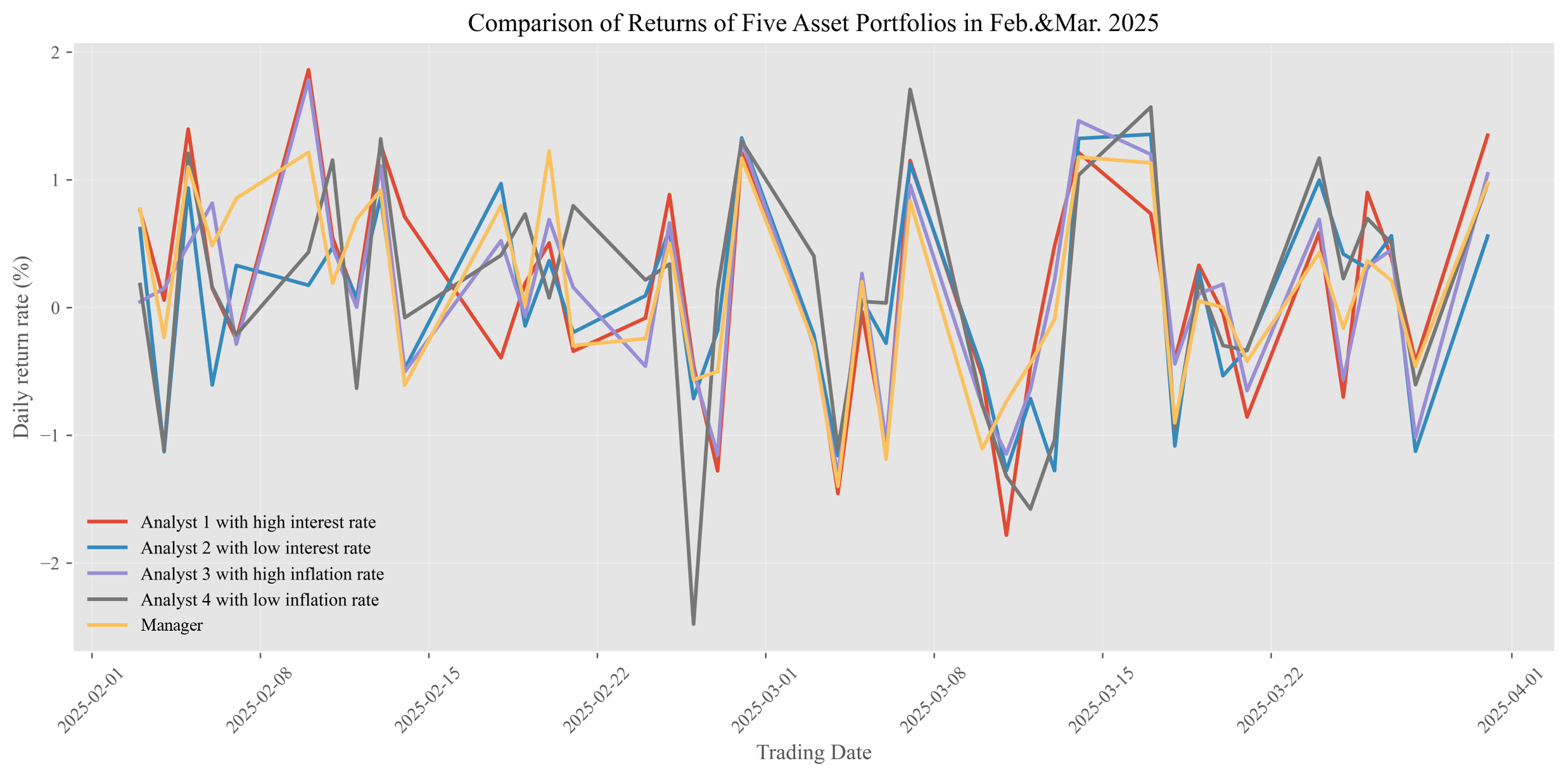}
        \caption{Recession regime}
        \label{fig:recession}
    \end{subfigure}
    \begin{subfigure}[b]{0.9\textwidth}
        \centering
        \includegraphics[width=\textwidth]{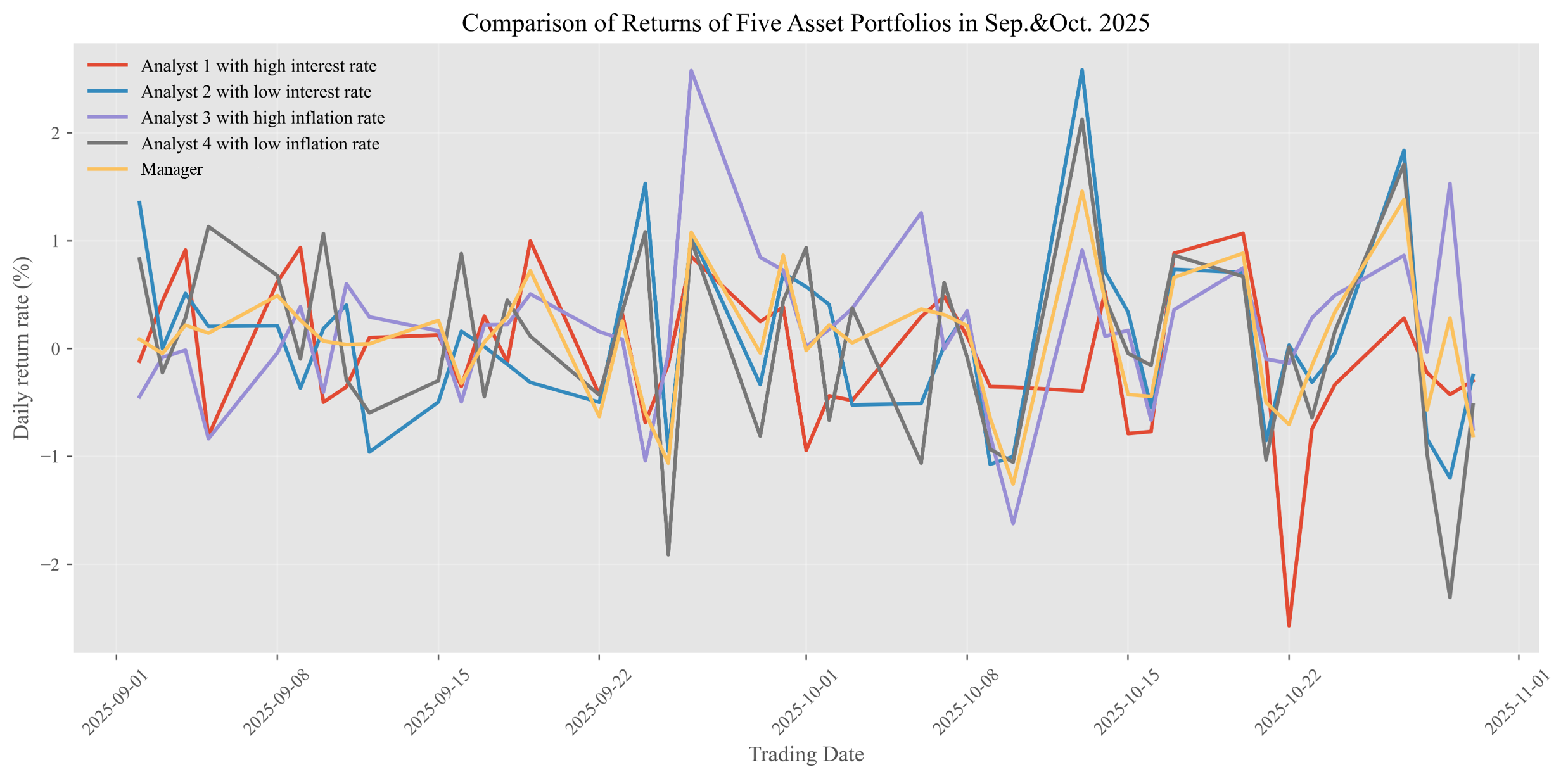}
        \caption{Expansion regime}
        \label{fig:expansion}
    \end{subfigure}
    \caption{Baseline Portfolio Daily Return.}
    \label{fig:portfolio_daily_return}
\end{figure}

\subsection{Sensitivity Analysis}

In this subsection, we conduct a sensitivity analysis to ensure our baseline findings are not artifacts of specific parameter choices. Table \ref{table: Change_Time_Window} reports portfolio performance under alternative time windows ($T=120$ and $T=180$). During the recession regime, the manager's portfolio is more noticeably dragged down by Analyst 2's misjudgment than in the baseline; crucially, however, it still significantly outperforms Analyst 2's isolated strategy. In the expansion regime, the integrated portfolio maintains mid-tier performance similar to the baseline case at $T=120$ but underperforms at $T=180$.

\begin{remark}
	Varying time windows induces performance fluctuations due to the bias-variance tradeoff inherent in historical data estimation: shorter windows increase estimation variance, whereas longer windows may incorporate outdated market fundamentals. 
    However, our primary goal here is not to evaluate the absolute robustness of ES itself, but rather to verify whether the manager's WGRM-based approach consistently maintains its structural advantage across diverse parameter specifications.
\end{remark}

\begin{table}[htbp]
    \centering
    \small
    \caption{Change Time Window.}
    \label{table: Change_Time_Window}
    \begin{threeparttable}
        \begin{adjustbox}{center}
        \begin{tabular}{lcccccc}
        \toprule
         & \multicolumn{3}{c}{Recession regime} & \multicolumn{3}{c}{Expansion regime} \\
        \cmidrule(lr){2-4} \cmidrule(lr){5-7}
        Portfolio & \begin{tabular}{@{}c@{}}Two-month\\return\end{tabular} & \begin{tabular}{@{}c@{}}Sharpe\\ratio\end{tabular} & \begin{tabular}{@{}c@{}}Sortino\\ratio\end{tabular} & \begin{tabular}{@{}c@{}}Two-month\\return\end{tabular} & \begin{tabular}{@{}c@{}}Sharpe\\ratio\end{tabular} & \begin{tabular}{@{}c@{}}Sortino\\ratio\end{tabular} \\
        \midrule
        \textbf{Panel A: Change time window to $T = 120$} \\
        \midrule
        Analyst 1 & 8.38\% & 2.910 & 5.480 & 9.85\% & 3.396 & 5.061 \\
        Analyst 2 & 0.71\% & 0.059 & 0.094 & 10.08\% & 4.321 & 9.064 \\
        Analyst 3 & 3.91\% & 1.675 & 3.529 & 6.93\% & 3.307 & 5.224 \\
        Analyst 4 & 5.64\% & 1.515 & 2.680 & 19.26\% & 5.974 & 16.135 \\
        Manager & 3.00\% & 1.266 & 2.695 & 7.42\% & 3.960 & 7.586 \\
        \midrule
        \textbf{Panel B: Change time window to $T = 180$} \\
        \midrule
        Analyst 1 & 5.20\% & 2.388 & 4.290 & -2.02\% & -1.518 & -2.208 \\
        Analyst 2 & -1.29\% & -0.748 & -1.079 & 2.27\% & 0.841 & 1.420 \\
        Analyst 3 & 4.68\% & 2.154 & 3.653 & 3.92\% & 2.067 & 3.239 \\
        Analyst 4 & 3.56\% & 1.172 & 1.937 & 1.24\% & 0.288 & 0.503 \\
        Manager & 0.26\% & -0.161 & -0.285 & -1.16\% & -0.913 & -1.292 \\
        \bottomrule
        \end{tabular}
        \end{adjustbox}
    \end{threeparttable}
\end{table}

Table \ref{table:Change_Quantile} reports the results when the ES level is adjusted to $\alpha=0.90$ and $\alpha=0.99$. The manager’s portfolio performs consistently with the baseline, except for notable underperformance during the expansion regime at $\alpha=0.99$. The observed underperformance might be attributed to excessive conservatism. This heightened risk aversion likely leads the optimization process to over-allocate to low-volatility, defensive assets. Such assets tend to underperform in bull markets, as they fail to capture the strong upward momentum typically associated with growth-oriented stocks during economic expansions.

\begin{table}[ht]
    \centering
    \small
    \caption{Change Level.}
    \label{table:Change_Quantile}
    \begin{threeparttable}
        \begin{adjustbox}{center}
        \begin{tabular}{lcccccc}
        \toprule
         & \multicolumn{3}{c}{Recession regime} & \multicolumn{3}{c}{Expansion regime} \\
        \cmidrule(lr){2-4} \cmidrule(lr){5-7}
        Portfolio & \begin{tabular}{@{}c@{}}Two-month\\return\end{tabular} & \begin{tabular}{@{}c@{}}Sharpe\\ratio\end{tabular} & \begin{tabular}{@{}c@{}}Sortino\\ratio\end{tabular} & \begin{tabular}{@{}c@{}}Two-month\\return\end{tabular} & \begin{tabular}{@{}c@{}}Sharpe\\ratio\end{tabular} & \begin{tabular}{@{}c@{}}Sortino\\ratio\end{tabular} \\
        \midrule
        \textbf{Panel A: Change level to $\alpha=0.90$} \\
        \midrule
        Analyst 1 & 1.30\% & 0.363 & 0.686 & 1.01\% & 0.223 & 0.324 \\
        Analyst 2 & 1.42\% & 0.435 & 0.884 & 11.43\% & 5.960 & 12.281 \\
        Analyst 3 & 0.04\% & -0.270 & -0.534 & 9.70\% & 5.252 & 10.859 \\
        Analyst 4 & 4.63\% & 1.757 & 2.523 & 2.43\% & 0.853 & 1.515 \\
        Manager & 2.63\% & 1.078 & 2.153 & 8.01\% & 4.579 & 9.388 \\
        \midrule
        \textbf{Panel B: Change level to $\alpha=0.99$} \\
        \midrule
        Analyst 1 & 6.70\% & 2.715 & 5.529 & -2.78\% & -1.866 & -2.565 \\
        Analyst 2 & -0.19\% & -0.382 & -0.869 & 3.70\% & 1.256 & 2.546 \\
        Analyst 3 & 6.96\% & 2.977 & 6.180 & 9.91\% & 5.073 & 11.733 \\
        Analyst 4 & 4.56\% & 1.723 & 2.584 & 0.27\% & -0.138 & -0.216 \\
        Manager & 7.88\% & 3.927 & 7.604 & -2.62\% & -1.471 & -2.618 \\
        \bottomrule
        \end{tabular}
        \end{adjustbox}
    \end{threeparttable}
\end{table}

Table \ref{table:Change_Target_Return} displays the results with target returns scaled by 0.5 and 2. The manager's portfolio aligns with the baseline except during the expansion regime with a doubled target, unexpectedly posting a negative return while all individual analysts remain profitable. This counterintuitive result highlights the tension between aggressive return targets and the WGRM framework's mandate for risk mitigation. Pressured to pursue outsized gains in a market already exceeding historical averages, the integrated method's requirement to balance heterogeneous assessments structurally prevents over concentration in high-momentum assets. Consequently, the optimizer defaults to less volatile, value-oriented stocks, which inherently lag in a momentum-driven bull market.

\begin{table}[ht]
    \centering
    \small
    \caption{Change Target Return.}
    \label{table:Change_Target_Return}
    \begin{threeparttable}
        \begin{adjustbox}{center}
        \begin{tabular}{lcccccc}
        \toprule
         & \multicolumn{3}{c}{Recession regime} & \multicolumn{3}{c}{Expansion regime} \\
        \cmidrule(lr){2-4} \cmidrule(lr){5-7}
        Portfolio & \begin{tabular}{@{}c@{}}Two-month\\return\end{tabular} & \begin{tabular}{@{}c@{}}Sharpe\\ratio\end{tabular} & \begin{tabular}{@{}c@{}}Sortino\\ratio\end{tabular} & \begin{tabular}{@{}c@{}}Two-month\\return\end{tabular} & \begin{tabular}{@{}c@{}}Sharpe\\ratio\end{tabular} & \begin{tabular}{@{}c@{}}Sortino\\ratio\end{tabular} \\
        \midrule
        \textbf{Panel A: Change target return to $\theta \times 2$} \\
        \midrule
        Analyst 1 & 7.68\% & 3.156 & 5.752 & 3.83\% & 1.370 & 2.160 \\
        Analyst 2 & 1.18\% & 0.285 & 0.497 & 7.89\% & 2.326 & 5.134 \\
        Analyst 3 & 7.11\% & 2.501 & 5.829 & 18.24\% & 5.470 & 8.988 \\
        Analyst 4 & 4.52\% & 1.876 & 3.623 & 2.92\% & 0.914 & 1.495 \\
        Manager & 7.98\% & 3.568 & 7.228 & -1.70\% & -0.937 & -1.367 \\
        \midrule
        \textbf{Panel B: Change target return to $\theta / 2$} \\
        \midrule
        Analyst 1 & 6.70\% & 2.715 & 5.529 & -1.49\% & -1.114 & -1.407 \\
        Analyst 2 & -0.16\% & -0.369 & -0.837 & 1.80\% & 0.507 & 0.840 \\
        Analyst 3 & 8.74\% & 3.811 & 7.434 & 8.23\% & 4.998 & 9.160 \\
        Analyst 4 & 4.56\% & 1.723 & 2.584 & 0.71\% & 0.037 & 0.052 \\
        Manager & 7.47\% & 3.762 & 7.699 & 7.24\% & 3.814 & 6.165 \\
        \bottomrule
        \end{tabular}
        \end{adjustbox}
    \end{threeparttable}
\end{table}

Table \ref{table:Change_Underlying_Assets} replicates the baseline and sensitivity analyses using S\&P 500 constituents. Outcomes consistently match, and occasionally surpass, the initial findings. In several cases (e.g., Panels A, D, E, and G), during the recession regime, the manager's portfolio maintains a positive return even in instances where every individual analyst's strategy incurs losses. This demonstrates the resilience of the WGRM framework across different asset universes. Furthermore, compared to the tech-heavy NASDAQ 100, the broader sectoral coverage of the S\&P 500 inherently introduces more defensive and value-oriented stocks. This expanded asset base provides more options for diversification, and also confirms that the WGRM method's structural robustness is not asset-specific, highlighting its superior adaptability to diverse market segments.

\begin{table}[htbp]
    \centering
    \small
    \caption{Change Underlying Assets.}
    \label{table:Change_Underlying_Assets}
    \begin{threeparttable}
        \begin{adjustbox}{center}
        \begin{tabular}{lcccccc}
        \toprule
         & \multicolumn{3}{c}{Recession regime} & \multicolumn{3}{c}{Expansion regime} \\
        \cmidrule(lr){2-4} \cmidrule(lr){5-7}
        Portfolio & \begin{tabular}{@{}c@{}}Two-month\\return\end{tabular} & \begin{tabular}{@{}c@{}}Sharpe\\ratio\end{tabular} & \begin{tabular}{@{}c@{}}Sortino\\ratio\end{tabular} & \begin{tabular}{@{}c@{}}Two-month\\return\end{tabular} & \begin{tabular}{@{}c@{}}Sharpe\\ratio\end{tabular} & \begin{tabular}{@{}c@{}}Sortino\\ratio\end{tabular} \\
        \midrule
        
        \textbf{Panel A: Benchmark case} \\
        \midrule
        Analyst 1 & -1.84\% & -0.839 & -1.089 & 4.28\% & 2.409 & 4.745 \\
        Analyst 2 & -0.06\% & -0.430 & -0.704 & 3.20\% & 1.212 & 2.196 \\
        Analyst 3 & -2.11\% & -1.461 & -2.023 & 7.86\% & 4.332 & 5.640 \\
        Analyst 4 & -4.20\% & -1.893 & -3.147 & 4.39\% & 1.876 & 3.361 \\
        Manager & 0.44\% & -0.144 & -0.186 & 3.44\% & 1.427 & 3.011 \\
        \midrule
        
        \textbf{Panel B: Change time window to $T=120$} \\
        \midrule
        Analyst 1 & 1.03\% & 0.163 & 0.237 & 7.55\% & 2.923 & 4.936 \\
        Analyst 2 & -1.04\% & -3.781 & -6.903 & 10.76\% & 4.608 & 7.351 \\
        Analyst 3 & -0.90\% & -0.571 & -0.905 & 7.86\% & 4.332 & 5.640 \\
        Analyst 4 & -10.43\% & -2.578 & -4.340 & 8.28\% & 2.750 & 4.637 \\
        Manager & 0.09\% & -0.429 & -0.529 & 8.73\% & 4.934 & 7.010 \\
        \midrule
        
        \textbf{Panel C: Change time window to $T=180$} \\
        \midrule
        Analyst 1 & 0.67\% & 0.038 & 0.055 & 3.14\% & 1.278 & 1.881 \\
        Analyst 2 & -1.99\% & -1.556 & -2.499 & 1.39\% & 0.378 & 0.769 \\
        Analyst 3 & -0.14\% & -0.336 & -0.497 & 5.55\% & 3.139 & 4.103 \\
        Analyst 4 & -2.03\% & -1.407 & -2.443 & 4.37\% & 2.037 & 4.150 \\
        Manager & -0.19\% & -0.401 & -0.662 & 3.65\% & 1.724 & 4.287 \\
        \midrule
        
        \textbf{Panel D: Change level to $\alpha=0.90$} \\
        \midrule
        Analyst 1 & -1.84\% & -0.839 & -1.089 & 4.28\% & 2.409 & 4.745 \\
        Analyst 2 & -0.06\% & -0.430 & -0.704 & 4.71\% & 1.856 & 3.117 \\
        Analyst 3 & -2.11\% & -1.461 & -2.023 & 7.86\% & 4.332 & 5.640 \\
        Analyst 4 & -4.20\% & -1.893 & -3.147 & 3.98\% & 2.009 & 4.771 \\
        Manager & 0.64\% & 0.069 & 0.088 & 3.88\% & 1.919 & 4.769 \\
        \midrule
        
        \textbf{Panel E: Change level to $\alpha=0.99$} \\
        \midrule
        Analyst 1 & -1.84\% & -0.839 & -1.089 & 4.28\% & 2.409 & 4.745 \\
        Analyst 2 & -0.06\% & -0.430 & -0.704 & 3.20\% & 1.212 & 2.196 \\
        Analyst 3 & -2.11\% & -1.461 & -2.023 & 7.86\% & 4.332 & 5.640 \\
        Analyst 4 & -4.20\% & -1.893 & -3.147 & 4.39\% & 1.876 & 3.361 \\
        Manager & 0.44\% & -0.144 & -0.186 & 2.48\% & 0.893 & 1.566 \\
        \midrule
        
        \textbf{Panel F: Change target return to $\theta \times 2$} \\
        \midrule
        Analyst 1 & -1.84\% & -0.839 & -1.089 & 7.38\% & 2.843 & 5.052 \\
        Analyst 2 & -0.06\% & -0.430 & -0.704 & 2.19\% & 0.504 & 0.722 \\
        Analyst 3 & -2.03\% & -1.367 & -1.903 & 4.85\% & 2.335 & 3.202 \\
        Analyst 4 & -4.20\% & -1.893 & -3.147 & 3.04\% & 0.837 & 1.174 \\
        Manager & -2.20\% & -1.548 & -2.140 & 6.27\% & 1.664 & 3.766 \\
        \midrule
        
        \textbf{Panel G: Change target return to $\theta / 2$} \\
        \midrule
        Analyst 1 & -1.84\% & -0.839 & -1.089 & 2.35\% & 1.238 & 2.270 \\
        Analyst 2 & -0.06\% & -0.430 & -0.704 & 6.17\% & 2.744 & 6.208 \\
        Analyst 3 & -2.11\% & -1.461 & -2.023 & 7.86\% & 4.332 & 5.640 \\
        Analyst 4 & -4.20\% & -1.893 & -3.147 & 3.70\% & 1.556 & 3.736 \\
        Manager & 1.58\% & 0.739 & 0.947 & 3.30\% & 1.417 & 3.635 \\
        \bottomrule
        \end{tabular}
        \end{adjustbox}
    \end{threeparttable}
\end{table}

In summary, our empirical findings confirm that WGRM-based portfolios exhibit greater conservatism and robustness. While this framework may not generate outstanding raw returns in expansion regimes due to its inherent focus on downside protection, it consistently delivers superior resilience during recessions. Most crucially, it acts as a structural buffer, neutralizing the idiosyncratic risk of isolated misjudgments and preventing the severe consequences of over-relying on a single, flawed market outlook.

\section*{Acknowledgements}
YL acknowledges financial support from the National Natural Science Foundation of China (Grant No. 12401624), The Chinese University of Hong Kong (Shenzhen) University Development Fund (Grant No. UDF01003336) and Shenzhen Science and Technology Program (Grant No. RCBS20231211090814028, JCYJ20250604141203005, 2025TC0010) and is partly supported by the Guangdong Provincial Key Laboratory of Mathematical Foundations for Artificial Intelligence (Grant No. 2023B1212010001). 
YW acknowledges financial support from the Natural Sciences and Engineering Research Council of Canada (RGPIN-2023-04674, DGECR-2023-00454), and the start-up fund at Carleton University. 
YW thanks 
The Chinese University of Hong Kong (Shenzhen) for the kind hospitality during her visit in 2025. 
\bibliographystyle{apalike}
\bibliography{refer}

\appendix
\section*{Appendix: Proofs}

\begin{proof}[\textbf{Proof of Theorem \ref{theo:WeightedGRM_condition_sup}}]
\label{proof:WeightedGRM_condition_sup}
    The ``if" statement can be checked directly. We only focus on the ``only if" statement.
    
    (1) \textbf{Step 1: Dual representation via an auxiliary convex function}. We first define an auxiliary function:
    \begin{align}
    	g(\Phi_{\calQ, X})
    	:&= 
    	f(\Phi_{\calQ, X}) + \delta(\Phi_{\calQ, X}|\mathcal{C}) \in (-\infty,\infty],
    	\quad
    	\Phi_{\calQ, X} \in \mathbb{R}^{n},
        \label{eq:g}
        \\
    	\delta(\Phi_{\calQ, X}|\mathcal{C})
    	&=
    	\begin{cases}
            0, & \text{if}\,\, \Phi_{\calQ, X} \in \mathcal{C}, \\
            \infty, & \text{if}\,\, \Phi_{\calQ, X} \notin \mathcal{C}.
        \end{cases}
        \label{eq:delta}
    \end{align}
    Obviously, $g(\Phi_{\calQ, X})=f(\Phi_{\calQ, X})$ on $\mathcal{C}$. Since $f(\Phi_{\calQ, X})$ satisfies (B1)-(B2), it is easy to verify that $g(\Phi_{\calQ, X})$ also satisfies positive homogeneity, translation invariance, and monotonicity.
    Moreover, we can further verify that $g(\Phi_{\calQ, X})$ is also sub-additive, i.e., $g(\Phi_{\calQ, X}+\Phi_{\calQ, Y}) \leq g(\Phi_{\calQ, X}) + g(\Phi_{\calQ, Y})$, for any $\Phi_{\calQ, X}, \Phi_{\calQ, Y} \in \mathbb{R}^n$. Note that if $\Phi_{\calQ, X}, \Phi_{\calQ, Y} \in \mathcal{C}$, then  the inequality is valid by the comonotonic sub-additivity of $f$. If $\Phi_{\calQ, X}$ or $\Phi_{\calQ, Y} \notin \mathcal{C}$, then the right-hand side of Eq.(\ref{eq:g}) goes to infinity due to $\delta(\Phi_{\calQ, X}|\mathcal{C})$. The positive homogeneity and sub-additivity indicate convexity of $g$.
    Furthermore, since $f(\cdot) < \infty$, we have $\text{dom}\, g :=\{\Phi_{\calQ, X} \in \mathbb{R}^{n} \mid g(\Phi_{\calQ, X}) < \infty\} \neq \emptyset$, i.e., $g(\cdot)$ is proper. From the fact that $f(\cdot)$ is Lipschitz continuous with  respect to the maximum-norm $\|\cdot\|_{\infty}$ and $\mathcal{C}$ is a closed convex set, $g(\cdot)$ is lower semi-continuous (l.s.c.). By Fenchel-Moreau Biconjugation Theorem (Chapter 12, \cite{Rockafellar1970}), it holds that
    \begin{align}
    \label{eq:dual function}
    	g(\Phi_{\calQ, X})
    	=
    	\sup \limits_{\mu_{\mathcal{Q}} \in \mathbb{R}^{n}} \{\left\langle \mu_{\mathcal{Q}},\Phi_{\calQ, X} \right\rangle - g^{*}(\mu_{\mathcal{Q}})\},
    	\quad
    	\Phi_{\calQ, X} \in \mathbb{R}^{n},
    \end{align}
    where $g^{*}:\mathbb{R}^{n} \to (-\infty,\infty]$, $g^{*}(\mu_{\mathcal{Q}}) = \sup\limits_{\Phi_{\calQ, X} \in \mathbb{R}^{n}} \{ \left\langle \mu_{\mathcal{Q}},\Phi_{\calQ, X} \right\rangle - g(\Phi_{\calQ, X})\}$, for $\mu_{\mathcal{Q}} =(\mu_{1},\mu_{2},\dots,\mu_{n})^{T} \in \mathbb{R}^{n}$, is the dual function of $g(\cdot)$.

    \textbf{Step 2: Structure of the conjugate and characterization of the weight set}.
    Next we investigate the characteristic of $\text{dom}\, g^{*}$. First, since $\forall\,k \in \mathbb{R}$, $k \mathbf{1} \in \mathcal{C}$, by translation invariance, we have
    \begin{align*}
    	g^{*}(\mu_{\mathcal{Q}}) 
    	&=
    	\sup\limits_{\Phi_{\calQ, X} \in \mathbb{R}^{n}} \{ \left\langle \mu_{\mathcal{Q}},\Phi_{\calQ, X} \right\rangle - g(\Phi_{\calQ, X})\}
    	\geq
    	\sup\limits_{k \in \mathbb{R}} \{ \left\langle \mu_{\mathcal{Q}},k \mathbf{1} \right\rangle - g(k \mathbf{1})\}
    	=
    	\{ \left\langle \mu_{\mathcal{Q}},\mathbf{1} \right\rangle - g(\mathbf{1})\} \cdot \sup\limits_{k \in \mathbb{R}} k.
    \end{align*}
    If $\left\langle \mu_{\mathcal{Q}},\mathbf{1} \right\rangle - g(\mathbf{1})=\left\langle \mu_{\mathcal{Q}},\mathbf{1} \right\rangle - f(\mathbf{1})=0$, i.e., $\sum_{i=1}^{n} \mu_{i}=1$, we have $g^{*}(\mu_{\mathcal{Q}}) \geq 0$. Alternatively, if $\left\langle \mu_{\mathcal{Q}},\mathbf{1} \right\rangle - g(\mathbf{1})\neq 0$, i.e., $\sum_{i=1}^{n} \mu_{i} \neq 1$, we have $g^{*}(\mu_{\mathcal{Q}}) = \infty$, since $k$ is arbitrary. This indicates that $g(\mu_{\mathcal{Q}})$ is non-negative and that $\text{dom}\, g^{*} \subseteq \{ \mu_{\mathcal{Q}} \in \mathbb{R}^{n} \,|\, \sum^{n}_{i=1} \mu_{i}=1 \}$. Second, by positive homogeneity,
    \begin{align*}
    	g^{*}(\mu_{\mathcal{Q}}) &= \sup\limits_{\Phi_{\calQ, X} \in \mathbb{R}^{n}} \{ \left\langle \mu_{\mathcal{Q}},\Phi_{\calQ, X} \right\rangle - g(\Phi_{\calQ, X})\}
    	=
        \sup\limits_{\Phi_{\calQ, X} \in \mathbb{R}^{n}} a\ \{ \frac{1}{a} \left\langle \mu_{\mathcal{Q}},\Phi_{\calQ, X} \right\rangle - \frac{1}{a} g(\Phi_{\calQ, X})\}\\
        &=
        a \sup\limits_{\Phi_{\calQ, X}/a\, \in \mathbb{R}^{n}} \{  \langle \mu_{\mathcal{Q}},\frac{1}{a}\Phi_{\calQ, X} \rangle -  g(\frac{1}{a}\Phi_{\calQ, X})\}
        =
        a g^{*}(\mu_{\mathcal{Q}}),
    \end{align*} 
    for arbitrary $a \geq 0$. This indicates that if $g^{*}(\mu_{\mathcal{Q}}) < \infty$, it must hold that $g^{*}(\mu_{\mathcal{Q}})=0$. Combining the above two claims, we can conclude that
    \begin{align}
    \label{eq:g*}
    	g^{*}(\mu_{\mathcal{Q}})=\delta(\cdot|\text{dom}\,g^{*})
        \quad \text{and} \quad
        g(\Phi_{\calQ, X})=\sup\limits_{\mu_{\mathcal{Q}} \in \text{dom}\,g^{*}} \left\langle \mu_{\mathcal{Q}},\Phi_{\calQ, X} \right\rangle,
    \end{align}
    where the second equation of (\ref{eq:g*}) yields from the fact that, when $\mu_{\mathcal{Q}} \notin \text{dom}\,g^{*}$, we have $g^{*}(\mu_{\mathcal{Q}}) =\infty$, implying that $g(\Phi_{\calQ, X}) = -\infty$.

    Define the set of maximizers at $\Phi_{\calQ, X}$ of Eq.(\ref{eq:dual function}) as $\partial g(\Phi_{\calQ, X})$, i.e., $\partial g(\Phi_{\calQ, X}):=\{ \mu_{\mathcal{Q}} \in \text{dom}\,g^{*} \mid g(\Phi_{\calQ, X}) + g^{*}(\mu_{\mathcal{Q}})= \left\langle \mu_{\mathcal{Q}},\Phi_{\calQ, X} \right\rangle \}$. 
    It can be shown that $\forall\, \Phi_{\calQ, X} \in \text{int}\mathcal{C}$, $\partial g(\Phi_{\calQ, X}) \neq \emptyset$. 
    Now, fix any $\Phi_{\calQ, X} \in \text{int}\mathcal{C}$ and let $\mu_{\mathcal{Q}} \in \partial g(\Phi_{\calQ, X})$. 
    From Eq.(\ref{eq:g*}), we know that $g(\Phi_{\calQ, X})=\left\langle \mu_{\mathcal{Q}},\Phi_{\calQ, X} \right\rangle$ holds. Let $\mathbf{e}_{1}, \mathbf{e}_{2},\dots,\mathbf{e}_{n}$ be the canonical basis of $\mathbb{R}^{n}$. 
    Since $\Phi_{\calQ, X} \in \text{int}\mathcal{C}$, for each $i$, there exists an $\epsilon >0$ small enough such that $\Phi_{\calQ, X}-\epsilon \mathbf{e}_{i} \in \mathcal{C}$. Note that $\mu_{\calQ}$ is an element of $\partial g(\Phi_{\calQ, X})$, but it does not necessarily indicate $\mu_{\mathcal{Q}} \in \partial g(\Phi_{\calQ, X}-\epsilon \mathbf{e}_{i})$ as well. Therefore, we have
    \begin{align}
    \label{proof:inequility}
        \left\langle \mu_{\mathcal{Q}},\Phi_{\calQ, X}-\epsilon \mathbf{e}_{i} \right\rangle
        \leq
        g(\Phi_{\calQ, X}-\epsilon \mathbf{e}_{i})
        \leq
        g(\Phi_{\calQ, X})
        =
        \left\langle \mu_{\mathcal{Q}},\Phi_{\calQ, X} \right\rangle,
    \end{align}
    where the second inequality of (\ref{proof:inequility}) results from monotonicity of $g$. Equivalently, $\left\langle \mu_{\mathcal{Q}},-\epsilon \mathbf{e}_{i} \right\rangle \leq 0$, i.e., $-\mathbf{e}_{i} \mu_{i} \leq 0$, $\mu_{i} \geq 0$, which indicates that all coordinates of $\mu_{\calQ}$ are non-negative. Since $g^{*}$ is l.s.c. (as $g$ is l.s.c.), the set $\mathcal{W}_{1}:= \text{dom\,} g^{*} \subseteq \mathcal{D}$ is closed and convex. To sum up, so far we have obtained that 
    \begin{align*}
        g(\Phi_{\calQ, X})=\sup\limits_{\mu_{\mathcal{Q}} \in \mathcal{W}_{1}}\left\langle \mu_{\mathcal{Q}},\Phi_{\calQ, X} \right\rangle,
        \quad
        \forall\, \Phi_{\calQ, X} \in \text{int}\mathcal{C}.
    \end{align*}

    \textbf{Step 3: Extension to the boundary and final representation}.
    For any boundary point $\Phi_{\calQ, X}$ of $\mathcal{C}$, we could choose a sequence $\{ \Phi_{\calQ, X}^{k}\} \subset \text{int}\mathcal{C}$ converging to $\Phi_{\calQ, X}$. Then by the Lipschitz continuity of $f$ and the fact that $g(\Phi_{\calQ, X})=f(\Phi_{\calQ, X})$ on $\mathcal{C}$, we have $f(\Phi_{\calQ, X})=\sup\limits_{\mu_{\mathcal{Q}} \in \mathcal{W}_{1}}\left\langle \mu_{\mathcal{Q}},\Phi_{\calQ, X} \right\rangle$.
    Finally, since $f(\Phi_{\calQ, X})$ is permutation invariant and $\Phi_{\calQ, X}^{q} \in \mathcal{C}$ for every $\Phi_{\calQ, X} \in \mathbb{R}^{n}$, we have
    \begin{align*}
        f(\Phi_{\calQ, X})
        =
        f(\Phi_{\calQ, X}^{q})
        =
        g(\Phi_{\calQ, X}^{q})
        =
        \sup\limits_{\mu_{\mathcal{Q}} \in \mathcal{W}_{1}}\left\langle \mu_{\mathcal{Q}},\Phi_{\calQ, X}^{q} \right\rangle,
        \quad
        \forall\, \Phi_{\calQ, X} \in \mathbb{R}^{n}.
    \end{align*}

    (2) Since now $f(\Phi_{\calQ, X})$ is assumed to be sub-additive, it is directly convex along with properties of positive homogeneity, translation invariance, and monotonicity. From Part (1), we immediately have
    \begin{align*}
        f(\Phi_{\calQ, X}) = \sup\limits_{\mu_{\mathcal{Q}} \in \text{dom\,}f^{*}} \left\{\left\langle \mu_{\mathcal{Q}},\Phi_{\calQ, X} \right\rangle - f^{*}(\mu_{\mathcal{Q}})  \right\}, \quad \Phi_{\calQ, X} \in \mathbb{R}^{n},
        \end{align*}
    where 
    \begin{align*}
        f^{*}(\mu_{\mathcal{Q}})=\sup\limits_{\Phi_{\calQ, X} \in \mathbb{R}^{n}} \left\{\left\langle \mu_{\mathcal{Q}},\Phi_{\calQ, X} \right\rangle - f(\Phi_{\calQ, X})  \right\},
    \end{align*}
    and $\text{dom\,}f^{*} \subseteq \mathcal{D}$. We first note that for any $\pi\in S_n$, $\mu_{\mathcal{Q}}^{\pi} \in \text{dom}\,g^{*}$ and $\Phi_{\calQ, X} \in \mathbb{R}^{n}$, there exists a unique inverse permutation $\pi^{-1}$ (since permutations are bijections) for permuted $\mu_{\mathcal{Q}}^{\pi}$ such that $\left\langle \mu_{\mathcal{Q}}^{\pi},\Phi_{\calQ, X} \right\rangle = \left\langle \mu_{\mathcal{Q}},\Phi_{\calQ, X}^{\pi^{-1}} \right\rangle$. Therefore, it is evident that
    \begin{align}
    \label{proof:f*}
        f^{*}(\mu_{\mathcal{Q}}^{\pi}) 
        =
        \sup\limits_{\Phi_{\calQ, X} \in \mathbb{R}^{n}} \left\{\left\langle \mu_{\mathcal{Q}}^{\pi},\Phi_{\calQ, X} \right\rangle - f(\Phi_{\calQ, X})  \right\}
        =
        \sup\limits_{\Phi_{\calQ, X}^{\pi^{-1}} \in \mathbb{R}^{n}} \left\{\left\langle \mu_{\mathcal{Q}},\Phi_{\calQ, X}^{\pi^{-1}}\right\rangle - f(\Phi_{\calQ, X}^{\pi^{-1}})  \right\}
        =
        f^{*}(\mu_{\mathcal{Q}}).
    \end{align}
    The second equality of (\ref{proof:f*}) is valid due to the permutation invariance of $f(\Phi_{\calQ, X})$, i.e., $f(\Phi_{\calQ, X}) = f(\Phi_{\calQ, X}^{\pi^{-1}})$. The above relationship implies that $\text{dom}\,g^{*}$ is also permutation invariant. Moreover, since $f(\Phi_{\calQ, X}) = \sup\limits_{\mu_{\mathcal{Q}} \in \mathcal{W}_{1}}\left\langle \mu_{\mathcal{Q}},\Phi_{\calQ, X}^{q} \right\rangle$, where $\Phi_{\calQ, X}^{q}$ is non-decreasingly ordered, the Rearrangement Inequality implies that the supremum would be attained when $\mu_{\mathcal{Q}}$ is also non-decreasingly ordered. Therefore, $\mathcal{W}_{2} \subseteq \mathcal{D} \,\cap\, \mathcal{C}$.
\end{proof}
\begin{proof}[\textbf{Proof of Proposition \ref{prop:WeightedGRM_condition_unique}}]
\label{proof:WeightedGRM_condition_unique}
    The ``if" statement is direct. We only focus on the ``only if" statement.
    
    (1) Proof for $f: \mathcal{C} \to \mathbb{R}$ satisfying (B1), (B2), (B6).
    Since $f$ satisfies (B1) positive homogeneity and translation invariance, setting $a=0$ and $b=0$ yields $f(\mathbf{0}) = 0$. Note that the intersection of $\mathcal{C}$ and its negative cone satisfies $\mathcal{C} \cap -\mathcal{C} = \{c \mathbf{1} \mid c \in \mathbb{R}, \mathbf{1} = (1,1,\dots,1)^T \in \mathbb{R}^n\}$. For any $c \in \mathbb{R}$, the constant vectors $c \mathbf{1}$ and $-c \mathbf{1}$ are comonotonic. Then by (B6) comonotonic additivity, it holds that $0 = f(\mathbf{0}) = f(c \mathbf{1} - c \mathbf{1}) = f(c \mathbf{1}) + f(-c \mathbf{1})$. Combined with (B1) positive homogeneity, $f$ is fully homogeneous over $\mathbb{R}$ on $\mathcal{C}$. 
    Since $\mathcal{C}$ is a convex cone and $f$ satisfies comonotonic additivity and full homogeneity, $f$ is linear on $\mathcal{C}$. By the Hahn-Banach Extension Theorem, there exists a unique linear extension $\hat{f}: \mathbb{R}^n \to \mathbb{R}$ such that $\hat{f}(\Phi_{\calQ, X}) = f(\Phi_{\calQ, X})$ for all $\Phi_{\calQ, X} \in \mathcal{C}$.
    By the Riesz Representation Theorem, there exists a unique vector $\mu_{\mathcal{Q}}^{*} \in \mathbb{R}^n$ such that $\hat{f}(\Phi_{\calQ, X}) = \left\langle \mu_{\mathcal{Q}}^{*}, \Phi_{\calQ, X} \right\rangle$ for all $\Phi_{\calQ, X} \in \mathbb{R}^n$; hence $f(\Phi_{\calQ, X}) = \left\langle \mu_{\mathcal{Q}}^{*}, \Phi_{\calQ, X} \right\rangle$ for all $\Phi_{\calQ, X} \in \mathcal{C}$.
    Finally, (B1) translation invariance enforces $\left\langle \mu_{\mathcal{Q}}^{*}, \mathbf{1} \right\rangle = 1$, and (B2) monotonicity requires $\mu_{i}^{*} \geq 0$ for all $i$. Thus, we have $\mu_{\mathcal{Q}}^{*} \in \mathcal{D}$.

    (2) Proof for $f: \mathbb{R}^n \to \mathbb{R}$ satisfying (B1), (B2), (B4), (B6).
    From Part (1), for any ordered risk vector $\Phi_{\calQ, X}^{q} \in \mathcal{C}$, there exists a unique $\mu_{\mathcal{Q}}^{*} \in \mathcal{D}$ such that $f(\Phi_{\calQ, X}^{q}) = \left\langle \mu_{\mathcal{Q}}^{*}, \Phi_{\calQ, X}^{q} \right\rangle$.
    For any unordered risk vector $\Phi_{\calQ, X} \notin \mathcal{C}$, (B4) permutation invariance guarantees that $f(\Phi_{\calQ, X}) = f(\Phi_{\calQ, X}^{\pi})$ for any permutation $\pi \in S_n$. By definition, there exists a permutation $\pi^*$ such that $\Phi_{\calQ, X}^{\pi^{*}} = \Phi_{\calQ, X}^{q} \in \mathcal{C}$, where $\Phi_{\calQ, X}^{q}$ is the non-decreasingly sorted version of $\Phi_{\calQ, X}$. Thus, it holds that $f(\Phi_{\calQ, X}) = f(\Phi_{\calQ, X}^{q}) = \left\langle \mu_{\mathcal{Q}}^{*}, \Phi_{\calQ, X}^{q} \right\rangle.$ The uniqueness of $\mu_{\mathcal{Q}}^{*}$ follows directly from Part (1).

    (3) Proof for $f: \mathbb{R}^n \to \mathbb{R}$ satisfying (B1), (B2), (B6').
    (B6') additivity and (B1) positive homogeneity imply that $f$ is a linear functional on $\mathbb{R}^n$. By the Riesz Representation Theorem, there exists a unique $\mu_{\mathcal{Q}}^{*} \in \mathbb{R}^n$ such that $f(\Phi_{\calQ, X}) = \left\langle \mu_{\mathcal{Q}}^{*}, \Phi_{\calQ, X} \right\rangle$ for all $\Phi_{\calQ, X} \in \mathbb{R}^n$.
    As in Part (1), translation invariance and monotonicity enforce $\mu_{\mathcal{Q}}^{*} \in \mathcal{D}$.
    
    (4) Proof for $f: \mathbb{R}^n \to \mathbb{R}$ satisfying (B1), (B2), (B4), (B6').
    From Part (3), we have $f(\Phi_{\calQ, X}) = \left\langle \mu_{\mathcal{Q}}^{*}, \Phi_{\calQ, X} \right\rangle$ for all $\Phi_{\calQ, X} \in \mathbb{R}^n$, with $\mu_{\mathcal{Q}}^{*} \in \mathcal{D}$.
    (B4) permutation invariance requires that for any permutation $\pi \in S_n$ and any $\Phi_{\calQ, X} \in \mathbb{R}^n$, it holds that $\left\langle \mu_{\mathcal{Q}}^{*}, \Phi_{\calQ, X} \right\rangle = \left\langle \mu_{\mathcal{Q}}^{*}, \Phi_{\calQ, X}^{\pi} \right\rangle.$ Rewriting the right-hand side using permutation properties yields $\left\langle \mu_{\mathcal{Q}}^{*}, \Phi_{\calQ, X}^{\pi} \right\rangle = \left\langle \mu_{\mathcal{Q}}^{* \pi^{-1}}, \Phi_{\calQ, X} \right\rangle$, where $\mu_{\mathcal{Q}}^{* \pi^{-1}}$ denotes the weight vector obtained by permuting $\mu_{\mathcal{Q}}^{*}$ with the inverse permutation $\pi^{-1}$.
    For this equality to hold for all $\Phi_{\calQ, X} \in \mathbb{R}^n$, it is necessary that $\mu_{\mathcal{Q}}^{*} = \mu_{\mathcal{Q}}^{* \pi^{-1}}$ for all $\pi \in S_n$. The only permutation-invariant vector in $\mathcal{D}$ is the uniform weight vector, i.e., $\mu_{\mathcal{Q}}^{*} = \left( \frac{1}{n}, \frac{1}{n}, \dots, \frac{1}{n} \right)^T.$

\end{proof}

\begin{proof}[\textbf{Proof of Theorem \ref{theo:WeightedGRM_condition_sup_continuous}}]
\label{proof:WeightedGRM_condition_sup_continuous}
    The ``if" statement can be checked directly. We only focus on the ``only if" statement.
    
    (1) \textbf{Step 1: Dual representation via an auxiliary convex function}. Let $\tilde{\mathcal{C}}=\{\varphi_{X} \in L^{\infty}([0,1]) \mid \varphi_{X}\text{ is non-decreasing and left-continuous} \}$ be a cone of monotone functions. For any $\varphi_{X} \in L^{\infty}([0,1])$, its quantile function $\varphi_X^q$ clearly belongs to $\tilde{\mathcal{C}}$. We first define an auxiliary function $g:L^{\infty}([0,1]) \to (-\infty,\infty]$ via:
    \begin{align*}
    	g(\varphi_{X}) 
    	:&= 
    	f(\varphi_{X})+\delta(\varphi_{X}|\tilde{\mathcal{C}}),
    	\\
    	\delta(\varphi_{X}|\tilde{\mathcal{C}})
    	&=
    	\begin{cases}
            0, & \text{if}\,\, \varphi_{X} \in \tilde{\mathcal{C}}, \\
            \infty, & \text{if}\,\, \varphi_{X} \notin \tilde{\mathcal{C}}.
        \end{cases}
    \end{align*}
    
    The properness and convexity of $g(\varphi_{X})$ can be established similarly to the discrete case. To establish the $\sigma(L^\infty, L^1)$-lower semi-continuity of $g$, which is equivalent to the Fatou property, we consider a uniformly bounded sequence $\{\varphi_{X_n}\} \subseteq L^{\infty}([0,1])$ such that $\varphi_{X_n} \xrightarrow{a.e.} \varphi_{X}$. We aim to show $g(\varphi_{X}) \leq \liminf\limits_{n\to \infty} g(\varphi_{X_n})$. If $\liminf\limits_{n\to \infty} g(\varphi_{X_n}) = \infty$, the inequality holds trivially. Otherwise, if $\liminf\limits_{n\to \infty} g(\varphi_{X_n}) < \infty$, then there exists a subsequence $\{\varphi_{X_{n_k}}\} \subseteq \tilde{\mathcal{C}}$ such that $\lim\limits_{k\to\infty} g(\varphi_{X_{n_k}}) = \liminf\limits_{n\to \infty} g(\varphi_{X_n})$. Since $\tilde{\mathcal{C}}$ is closed under almost everywhere convergence (taking the left-continuous modification), the limit $\varphi_{X}$ also belongs to $\tilde{\mathcal{C}}$.
    Thus, it suffices to show that $f(\varphi_{X}) \leq \liminf\limits_{k\to \infty} f(\varphi_{X_{n_k}})$, which is ensured by (C5).

    \begin{remark}
    \label{remark:convergence_choice}
        We emphasize the condition $\varphi_{X_n} \xrightarrow{a.e.} \varphi_{X}$ rather than convergence in the $L^{\infty}$-norm to ensure the functional is compatible with the $\sigma(L^{\infty}, L^{1})$ topology. This continuity requirement forces the dual representation to lie strictly within $L^{1}([0,1])$, thereby eliminating the singular components (purely finitely additive measures) that are otherwise inherent in the general dual space $(L^{\infty})^{*}$.
    \end{remark}

    Since $g$ is proper, convex and l.s.c., by the Fenchel-Moreau Biconjugation Theorem (Chapter 12, \cite{Rockafellar1970}), it holds that
    \begin{align*}
        g(\varphi_{X})=\sup\limits_{\nu\in L^{1}([0,1])}\left\{\int^{1}_{0}\varphi_{X}(t)\nu(t)\mathrm{d}t-g^{*}(\nu)\right\},
    \end{align*}
    where
    \begin{align*}
        g^{*}(\nu)=\sup\limits_{\varphi_{X} \in L^{\infty}([0,1])} 
        \left\{\int^{1}_{0}\varphi_{X}(t)\nu(t)\mathrm{d}t-g(\varphi_{X}) \right\}.
    \end{align*}

    \textbf{Step 2: Structure of the conjugate and characterization of the weight set}.
    Similar to the discrete case, we investigate the characteristic of $\text{dom }g^{*}$. First, for any $k \in \mathbb{R}$, the constant function $\varphi_{X}(t) \equiv k\in L^{\infty}([0,1])$. By (C1) affine invariance, we have
    \begin{align*}
    	g^{*}(\nu)
    	&=
    	\sup\limits_{\varphi_{X} \in L^{\infty}([0,1])} \left\{\int^{1}_{0}\varphi_{X}(t)\nu(t)\mathrm{d}t-g(\varphi_{X})\right\}\\
    	&\geq
    	\sup\limits_{k\in \mathbb{R}} \left\{\int^{1}_{0}k\mathbf{1}_{[0,1]}\,\nu(t)\mathrm{d}t-g(k\mathbf{1}_{[0,1]})\right\}
    	=
    	\left\{\int^{1}_{0}\nu(t)\mathrm{d}t-1 \right\} \cdot \sup\limits_{k\in \mathbb{R}} k.
    \end{align*}
    If $\int^{1}_{0}\nu(t)\mathrm{d}t=1$, then we have $g^{*}(\nu)\geq0$. Otherwise, if $\int^{1}_{0}\nu(t)\mathrm{d}t\neq 1$,  then $g^{*}(\nu) = \infty$. This indicates $\text{dom}\,g^{*} \subseteq \{\nu\in L^{1}([0,1])|\int^{1}_{0}\nu(t)\mathrm{d}t=1 \}$. Moreover, (C1) affine invariance also indicates 
    \begin{align*}
    	g^{*}(\nu)
    	&=
    	\sup\limits_{\varphi_{X} \in L^{\infty}([0,1])} \left\{\int^{1}_{0}\varphi_{X}(t)\nu(t)\mathrm{d}t-g(\varphi_{X})\right\}\\
    	&=
    	k\sup\limits_{\varphi_{X} /k\in L^{\infty}([0,1])} \left\{\int^{1}_{0}\frac{\varphi_{X}(t)}{k}\nu(t)\mathrm{d}t-g\left(\frac{\varphi_{X}}{k}\right)\right\}
    	=kg^{*}(\nu),
    	\quad \forall\, k\in \mathbb{R_{+}},
    \end{align*}
    which indicates that if $g^{*}(\nu) < \infty$, it must hold that $g^{*}(\nu)=0$. That is, if $\int^{1}_{0}\nu(t)\mathrm{d}t=1$, we have $g^{*}(\nu)= 0$. By the above claims, it is evident that 
    \begin{align*}
    	g^{*}(\nu)=\delta(\cdot|\text{dom}\,g^{*}),
    	\quad
    	g(\varphi_{X})=\sup\limits_{\nu\in \text{dom}\,g^{*}}\int^{1}_{0}\varphi_{X}(t)\nu(t)\mathrm{d}t.
    \end{align*}

    We next go on to prove that for each valid $\nu(t)$, it holds that $\nu(t) \geq 0$ almost everywhere. Since $g(\varphi_{X})$ is proper, convex and l.s.c., for any $\varphi_{X}$, it is easy to verify the non-emptiness of its subdifferential, i.e., $\partial g(\varphi_{X}) \neq \emptyset$. This indicates that for any $\varphi_{X} \in \tilde{\mathcal{C}}$, there exists $\nu \in \text{dom}\,g^{*}$ such that
    \begin{align}
    \label{eq:subgradient}
        \forall\,\varphi \in L^{\infty}([0,1]),
        \quad
        g(\varphi) \geq g(\varphi_{X})+\int^{1}_{0}(\varphi(t)-\varphi_{X}(t))\nu(t)\mathrm{d}t.
    \end{align}

    We first define another auxiliary set
    \begin{align*}
        \tilde{\mathcal{C}}_{0}= \left\{\varphi_{X} \in \tilde{\mathcal{C}}\mid \exists\, \delta >0 \text{\,such that\,} \varphi_{X}(t)-\varphi_{X}(s) \geq \delta(t-s),\forall\, 0\leq s<t\leq 1 \right\},
    \end{align*}
    which consists of strictly increasing functions with a uniform positive lower bound on the slope. We proceed by arising two lemmas.
    \begin{lemma}
    \label{lemma:smooth}
        Let $\varphi_{X} \in \tilde{\mathcal{C}}_{0}$ with $\varphi_{X}(t)-\varphi_{X}(s) \geq \delta(t-s),\forall\, 0\leq s<t\leq 1$. Let $\phi:[0,1] \to \mathbb{R}$ be a smooth function satisfying 
        
        (1) $\phi \geq 0$ on $[0,1]$,
        
        (2) $\phi$ is supported on the interval $[c,d] \subseteq (0,1)$,
        
        (3) $\|\phi \prime \|_{\infty} \leq K_{\phi}$ for some constant $K_{\phi}>0$.
        
        \noindent Then for all $0<\varepsilon < \delta / K_{\phi}$, the function $\varphi^{\varepsilon}(t)=\varphi_{X}(t)-\varepsilon \phi(t)$ belongs to $\tilde{\mathcal{C}}$.
    \end{lemma}
    \begin{proof}
        For any $s<t$, $\varphi^{\varepsilon}(t)-\varphi^{\varepsilon}(s)=[\varphi_{X}(t)-\varphi_{X}(s)] - \varepsilon [\phi(t)-\phi(s)]$. Since $|\phi(t)-\phi(s)| \leq K_{\phi}(t-s)$ by property (3) and $\varphi_{X}(t)-\varphi_{X}(s) \geq \delta(t-s)$ by definition, it holds that
        \begin{align*}
            \varphi^{\varepsilon}(t)-\varphi^{\varepsilon}(s) \geq \delta(t-s) - \varepsilon K_{\phi}(t-s) = (\delta - \varepsilon K_{\phi})(t-s)>0
        \end{align*}
        when $\varepsilon < \delta / K_{\phi}$. Hence $\varphi^{\varepsilon}$ is strictly increasing, and therefore $\varphi^{\varepsilon} \in \tilde{\mathcal{C}}$.
    \end{proof}
    \begin{lemma}
    \label{lemma:int}
        For any $\varphi_{X} \in \tilde{\mathcal{C}}_{0}$ and any $\nu \in \partial g(\varphi_{X}) \cap \text{dom\,} g^{*}$, we have $\nu(t) \geq 0$ a.e.
    \end{lemma}
    \begin{proof}
        Let $\varphi_{X} \in \tilde{\mathcal{C}}_{0}$ with minimum slope $\delta >0$. Fix any interval $[c,d] \subseteq (0,1)$, and let $\phi \geq 0$ be a smooth function supported on $[c,d]$ with $\|\phi \prime \|_{\infty} \leq K_{\phi}$. For $\varepsilon < \delta / K_{\phi}$, we define $\varphi^{\varepsilon} = \varphi_{X}- \varepsilon \phi$. By Lemma \ref{lemma:smooth}, we have $\varphi^{\varepsilon} \in \tilde{\mathcal{C}}$. Since $\nu \in \partial g(\varphi_{X})$, the subdifferential inequality in (\ref{eq:subgradient}) gives
        \begin{align*}
            g(\varphi^{\varepsilon}) \geq g(\varphi_{X}) + \int^{1}_{0} (\varphi^{\varepsilon}(t)-\varphi_{X}(t))\nu(t) \mathrm{d}t = g(\varphi_{X}) - \varepsilon \int^{1}_{0}\phi(t)\nu(t)\mathrm{d}t.
        \end{align*}
        Since $\varphi^{\varepsilon}(t)=\varphi_{X}(t)-\varepsilon \phi(t) \leq \varphi_{X}(t)$ and $\varphi^{\varepsilon},\varphi_{X} \in \tilde{\mathcal{C}}$, by (C2) pointwise monotonicity, we have $g(\varphi^{\varepsilon}) \leq g(\varphi_{X})$. Combining the above claims yields
        \begin{align*}
            \infty > g(\varphi_{X}) \geq g(\varphi^{\varepsilon}) \geq g(\varphi_{X}) - \varepsilon \int^{1}_{0} \phi(t) \nu(t) \mathrm{d}t,
        \end{align*}
        which implies $\varepsilon \int^{1}_{0} \phi(t) \nu(t) \mathrm{d}t \geq 0$, i.e., $\int^{1}_{0} \phi(t) \nu(t) \mathrm{d}t \geq 0$. Since this relationship holds true for any smooth function $\phi$ supported on any interval $[c,d] \subseteq (0,1)$, we conclude that $\nu(t) \geq 0$ holds almost everywhere on $(0,1)$, and hence almost everywhere on $[0,1]$.
    \end{proof}

    \textbf{Step 3: Extension to the boundary and final representation}.
    Since $g^{*}$ is l.s.c. (as $g$ is $\sigma(L^\infty, L^1)$-lower semi-continuous), the set $\mathcal{W}_{3}:= \text{dom\,} g^{*} \subseteq \{\nu \in L^{1}([0,1])|\int^{1}_{0}\nu(t)\mathrm{d}t=1, \nu(t)\geq 0\ \text{a.e.}\}$ is closed and convex. By Lemma \ref{lemma:int}, we immediately have
    \begin{align}
    \label{eq:G(g)}
        g(\varphi_{X})=f(\varphi_{X})=\sup\limits_{\nu \in \mathcal{W}_3} \int_{0}^{1} \varphi_{X}(t)\nu(t) \mathrm{d}t,
        \quad
        \forall\, \varphi_{X} \in \tilde{\mathcal{C}}_{0}.
    \end{align}
    For any function outside $\tilde{\mathcal{C}}_{0}$, i.e., $\varphi_{X} \in \tilde{\mathcal{C}} \setminus \tilde{\mathcal{C}}_{0}$, we can always choose a sequence $\{\varphi_{X_k}\} \in \tilde{\mathcal{C}}_{0}$ converging to $\varphi_{X}$. For example, choosing $\varphi_{X_k}(t)=\varphi_{X}(t) + t/k$ for $t\in[0,1]$, where $k\in \mathbb{N}$ is always valid. Note that for identity map $\text{id}(t)=t$, $t \in [0,1]$, it holds that $f(\text{id}) \leq f(\mathbf{1}_{[0,1]}) < \infty$. Then by the fact that
    \begin{align*}
        f(\varphi_{X}) \leq f(\varphi_{X_k}) \leq f(\varphi_{X}) + \frac{f(\text{id})}{k},
    \end{align*}
    we have $f(\varphi_{X_k}) \to f(\varphi_{X})$ when $k \to \infty$.
    Then the relationship Eq.(\ref{eq:G(g)}) is valid in the entire $\tilde{C}$.

    Finally, by the (C4) strong permutation invariance, we can extend the result outside $\tilde{\mathcal{C}}$:
    \begin{align*}
            f(\varphi_{X})=f(\varphi^{q}_{X})=g(\varphi^{q}_{X})=
    \sup\limits_{\nu\in \mathcal{W}_3}\int^{1}_{0}\varphi^{q}_{X}(t)\nu(t)\mathrm{d}t,
    \quad 
    \forall\, \varphi_X \in L^{\infty}([0,1]).
    \end{align*}

    (2) The idea is very similar as in the discrete case. By assumptions, it is straightforward to check that $f$ is proper and convex. The lower semi-continuity of $f$ is a corollary of Theorem 2.2 in \cite{Jouini2006}. Then from the above proof, we have 
    \begin{align*}
        f(\varphi_X) = \sup_{\nu \in \text{dom }f^{*}}
        \left\{\int_{0}^{1} \nu(t)\varphi_X(t)\mathrm{d}t  - f^{*}(\nu)\right\},
    \end{align*}
    where
    \begin{align*}
        f^{*}(\nu) = \sup_{ \varphi_X \in L^{\infty}([0,1])}
        \left\{\int_{0}^{1} \nu(t)\varphi_X(t) \mathrm{d}t - f(\varphi_X)\right\},
    \end{align*}
    and $\text{dom}\,f^{*} \subseteq \{\nu \in L^{1}([0,1])\mid \int^{1}_{0}\nu(t)\mathrm{d}t=1, \nu(t)\geq 0\ \text{a.e.}\}$.
    Note that for any measure-preserving transformations $T$, $\nu \in \text{dom}\, f^{*}$ and $\varphi_X \in L^{\infty}([0,1])$, there must exist an inverse measure-preserving transformations $T^{-1}$ such that $\int_{0}^{1} (\nu\circ T(t))\varphi_X(t) \mathrm{d}t=\int_{0}^{1} \nu(t)(\varphi_X\circ T^{-1}(t)) \mathrm{d}t$, which directly implies that 
    \begin{align*}
        f^*(\nu \circ T)
        &=
        \sup_{ \varphi_X \in L^{\infty}([0,1])}
        \left\{\int_{0}^{1} (\nu\circ T(t))\varphi_X(t) \mathrm{d}t - f(\varphi_X)\right\} \\
        &=
        \sup_{ \varphi_X \circ T^{-1} \in L^{\infty}([0,1])}
        \left\{\int_{0}^{1} \nu(t)(\varphi_X \circ T^{-1}(t)) \mathrm{d}t - f(\varphi\circ T^{-1})\right\}
        =
        f^*(\nu).
    \end{align*}
    The second equality is valid due to the strong permutation invariance of $f$. The above relationship implies that $\text{dom }f^{*}$ also satisfies strong permutation invariance. Thus it is reasonable to restrict $\mathcal{W}_4 := \text{dom }f^{*} \subseteq \{\nu \in L^{\infty}([0,1]) \mid \int^{1}_{0}\nu(t)\mathrm{d}t=1, \nu(t)\geq 0\ \text{a.e.},\ \nu \text{ is non-decreasing}\}$.

\end{proof}
\begin{proof}[\textbf{Proof of Proposition \ref{prop:WeightedGRM_condition_unique_continuous}}]
    The ``if" statement is direct. We only focus on the ``only if" statement.
    
    (1) Proof for $f:\hat{\mathcal{C}}\to \mathbb{R}$ satisfying (D1), (D2), (D3), and (D5).
    By (D3) comonotonic additivity, we have $f(\varphi_{X_1} + \varphi_{X_2}) = f(\varphi_{X_1}) + f(\varphi_{X_2})$ for any $\varphi_{X_1}, \varphi_{X_2} \in \hat{\mathcal{C}}$.

    Define $V = \text{span}(\hat{\mathcal{C}}) \subseteq L^{1}([0,1])$, the linear span of $\hat{\mathcal{C}}$. For any $\varphi = \varphi_{X_1} - \varphi_{X_2}$ with $\varphi_{X_1}, \varphi_{X_2} \in \hat{\mathcal{C}}$, we define $\hat{f}(\varphi) = f(\varphi_{X_1}) - f(\varphi_{X_2})$. To verify that $\hat{f}$ is well-defined, suppose $\varphi_{X_1} - \varphi_{X_2} = \varphi_{X_3} - \varphi_{X_4}$, i.e., $\varphi_{X_1} + \varphi_{X_4} = \varphi_{X_2} + \varphi_{X_3}$. By (D3) comonotonic additivity, we have
    \begin{align*}
        f(\varphi_{X_1}) + f(\varphi_{X_4}) = f(\varphi_{X_1} + \varphi_{X_4}) = f(\varphi_{X_2} + \varphi_{X_3}) = f(\varphi_{X_2}) + f(\varphi_{X_3}),
    \end{align*}
    which implies $f(\varphi_{X_1}) - f(\varphi_{X_2}) = f(\varphi_{X_3}) - f(\varphi_{X_4})$. Thus, $\hat{f}$ is uniquely defined on $V$.
    By definition, $\hat{f}$ is linear on $V$. With (D5) $L^1$-continuity, we have
    \begin{align*}
    |\hat{f}(\varphi)| = |f(\varphi_{X_1}) - f(\varphi_{X_2})| \leq M \|\varphi_{X_1} - \varphi_{X_2}\|_{L^{1}} = M \|\varphi\|_{L^{1}},
    \end{align*}
    confirming that $\hat{f}$ is bounded (hence continuous) on $V$.

    A critical observation is that by the Jordan Decomposition Theorem, any function of bounded variation can be expressed as the difference of two non-decreasing functions. Thus, the set of bounded variation functions is a subset of $V$. Since step functions (dense in $L^1([0,1])$) are of bounded variation, we know $V$ is dense in $L^1([0,1])$, i.e., $\bar{V} = L^{1}([0,1])$.

    As a continuous linear functional on a dense subspace, $\hat{f}$ admits a unique continuous linear extension $\tilde{f}: L^{1}([0,1]) \to \mathbb{R}$. By the Riesz Representation Theorem for $L^1$ Spaces, there exists a unique $\nu^{*} \in L^{\infty}([0,1])$ such that $\tilde{f}(\varphi_{X}) = \int_{0}^{1} \varphi_{X}(t) \nu^{*}(t)\mathrm{d} t, \forall\, \varphi_{X} \in L^{1}([0,1])$.
    Restricting to $\hat{\mathcal{C}}$, we have $f(\varphi_{X}) = \tilde{f}(\varphi_{X}) = \int_{0}^{1} g(t) \nu^{*}(t)\mathrm{d} t$.

    Finally, (D1) affine invariance implies $f(\mathbf{1}_{[0,1]}) = 1$, so $\int_{0}^{1} \nu^{*}(t)\mathrm{d} t = 1$. For any $0 \leq c < d \leq 1$, $\mathbf{1}_{(c,1]} \geq \mathbf{1}_{(d,1]}$ implies $f(\mathbf{1}_{(c,1]}) \geq f(\mathbf{1}_{(d,1]})$ by (D2) Pointwise monotonicity, and hence $\int_{c}^{d} \nu^{*}(t)\mathrm{d} t \geq 0$. The arbitrariness of $c$ and $d$ ensures $\nu^{*}(t) \geq 0$ a.e. on $[0,1]$.

    (2) Proof for $f:L^1([0,1]) \to \mathbb{R}$ satisfying (D1), (D2), (D3), (D4), and (D5).
    For any $\varphi_{X} \in L^{1}([0,1])$, there exists a Lebesgue measure-preserving transformation $T$ such that $\varphi_{X} = \varphi_{X}^{q} \circ T$. By (D4) strong permutation invariance, we have $f(\varphi_{X}) = f(\varphi_{X}^{q} \circ T) = f(\varphi_{X}^{q}).$
    Combining with the result from (1), we obtain $f(\varphi_{X}) = \int_{0}^{1} \varphi_{X}^{q}(t) \nu^{*}(t)\mathrm{d} t,\forall\, \varphi_{X} \in L^{1}([0,1]).$ Uniqueness of $\nu^{*}$ follows from (1).

    (3) Proof for $f:L^1([0,1]) \to \mathbb{R}$ satisfying (D1), (D2), (D3'), and (D5).
    (D1) affine invariance and (D3') additivity directly imply $f$ is a linear functional on $L^{1}([0,1])$. By (D5) $L^1$-continuity, $f$ is bounded (i.e., $\|f\|_{(L^1)^*} \leq M$). The Riesz Representation Theorem guarantees a unique $\nu^{*} \in L^{\infty}([0,1])$ such that $f(\varphi_{X}) = \int_{0}^{1} \varphi_{X}(t) \nu^{*}(t)\mathrm{d} t$, $\forall\, \varphi_{X} \in L^{1}([0,1]).$ Normalization ($\int_{0}^{1} \nu^{*}(t)\mathrm{d} t = 1$) and non-negativeness ($\nu^{*}(t) \geq 0$ a.e.) follow from (D1) and (D2), respectively, as shown in (1).

    (4) Proof for $f:L^1([0,1]) \to \mathbb{R}$ satisfying (D1), (D2), (D3'), (D4), and (D5). By Part (3) we know $f(\varphi_{X}) = \int_{0}^{1} \varphi_{X}(t) \nu^{*}(t)\mathrm{d} t$ holds for any $\varphi_{X} \in L^{1}([0,1])$. For any $\varphi_{X} \in L^{1}([0,1])$ and any measure-preserving transformation $T$, by (D4) strong permutation invariance, we have $\int_{0}^{1} \varphi_{X}(t) \nu^{*}(t)\mathrm{d}t=\int_{0}^{1} (\varphi_{X}\circ T(t)) \nu^{*}(t)\mathrm{d}t$. Note that for the measure-preserving transformation $T$, there must exist an inverse transformation $T^{-1}$ such that $\int_{0}^{1} (\varphi_{X}\circ T(t)) \nu^{*}(t)\mathrm{d}t=\int_{0}^{1} \varphi_{X}(t) (\nu^{*}\circ T^{-1}(t))\mathrm{d}t$, whcih further implies that $\int_{0}^{1} \varphi_{X}(t) \nu^{*}(t)\mathrm{d}t=\int_{0}^{1} \varphi_{X}(t) (\nu^{*}\circ T^{-1}(t))\mathrm{d}t$. Hence, we conclude that $\nu^{*}(t)=\nu^{*}(T^{-1}(t))$ for a.e. $t$ and any measure-preserving $T$, which implies that $\nu^{*}(t)$ is constant almost everywhere. The normalization of $\nu^{*}$ forces the constant to be $1$, i.e., $\nu^{*}(t) \equiv1$ a.e.

\end{proof}
\begin{proof}[\textbf{Proof of Theorem \ref{theo:multi_quadrangle}}]
We first prove that the measures in Eq.(\ref{eq:multi_quadrangle}) satisfy Eqs.(\ref{Erelation1})–(\ref{Erelation4}).
For Eq.(\ref{Erelation1}), we have
    \begin{align*}
    	\mathbb{E}_{\mathcal{Q}} [X] + \mathcal{D}_{\mathcal{Q}} (X) 
    	&=
    	\int_\mathcal{Q} \mathbb{E}_{P}[X] \mathrm{d} \mu_\mathcal{Q}(P) + \int_\mathcal{Q} \mathcal{D}_{P}(X) \mathrm{d} \mu_\mathcal{Q} (P)
    	=
    	\int_\mathcal{Q} (\mathbb{E}_{P}[X] + \mathcal{D}_{P}(X) )\mathrm{d} \mu_\mathcal{Q}(P) \\
    	&=
    	\int_\mathcal{Q} \mathcal{R}_{P}(X) \mathrm{d} \mu_\mathcal{Q}(P)
    	=
    	\mathcal{R}_{\mathcal{Q}}(X),
    \end{align*}
    \begin{align*}
        \mathcal{R}_{\mathcal{Q}}(X) - \mathbb{E}_{\mathcal{Q}}[X] 
    	&=
    	\int_\mathcal{Q} \mathcal{R}_{P}(X) \mathrm{d} \mu_\mathcal{Q} (P) -  \int_\mathcal{Q} \mathbb{E}_{P}[X] \mathrm{d} \mu_\mathcal{Q}(P)
    	=
    	\int_\mathcal{Q} (\mathcal{R}_{P}(X)- \mathbb{E}_{P}[X])  \mathrm{d} \mu_\mathcal{Q} (P) \\
    	&=
    	\int_\mathcal{Q} \mathcal{D}_{P}(X) \mathrm{d} \mu_\mathcal{Q}(P)
    	=
    	\mathcal{D}_{\mathcal{Q}}(X).
\end{align*}

For Eq.(\ref{Erelation2}), we have
    \begin{align*}
    	\mathbb{E}_{\mathcal{Q}} [X] + \mathcal{E}_{\mathcal{Q}} (X) 
    	&=
    	\int_\mathcal{Q} \mathbb{E}_{P}[X] \mathrm{d} \mu_\mathcal{Q}(P) + \min_{b(\mathcal{Q})} \left\{\int_\mathcal{Q} \mathcal{E}_{P}(X-b(P)) \mathrm{d} \mu_\mathcal{Q} (P)\mid \int_\mathcal{Q} b(P) \mathrm{d} \mu_\mathcal{Q}(P) = 0 \right\} \\
    	&=
    	\min_{b(\mathcal{Q})} \left\{\int_\mathcal{Q} (\mathcal{E}_{P}(X-b(P))+\mathbb{E}_{P}[X])  \mathrm{d} \mu_\mathcal{Q} (P)\mid \int_\mathcal{Q} b(P) \mathrm{d} \mu_\mathcal{Q}(P) = 0 \right\} \\
    	&=
    	\min_{b(\mathcal{Q})} \left\{\int_\mathcal{Q} (\mathcal{E}_{P}(X-b(P))+\mathbb{E}_{P}[X-b(P)]+b(P))  \mathrm{d} \mu_\mathcal{Q} (P)\mid \int_\mathcal{Q} b(P) \mathrm{d} \mu_\mathcal{Q}(P) = 0 \right\} \\
    	&=
    	\min_{b(\mathcal{Q})} \left\{\int_\mathcal{Q} (\mathcal{E}_{P}(X-b(P))+\mathbb{E}_{P}[X-b(P)])  \mathrm{d} \mu_\mathcal{Q} (P) +\int_\mathcal{Q} b(P)  \mathrm{d} \mu_\mathcal{Q} (P)\mid \int_\mathcal{Q} b(P) \mathrm{d} \mu_\mathcal{Q}(P) = 0 \right\} \\
    	&= 
    	\min_{b(\mathcal{Q})} \left\{\int_\mathcal{Q} \mathcal{V}_{P}(X-b(P))  \mathrm{d} \mu_\mathcal{Q} (P) \mid \int_\mathcal{Q} b(P) \mathrm{d} \mu_\mathcal{Q}(P) = 0 \right\} \\
    	&=\mathcal{V}_{\mathcal{Q}} (X),
    \end{align*}
    \begin{align*}
    	\mathcal{V}_{\mathcal{Q}} (X) - \mathbb{E}_{\mathcal{Q}} [X]
    	&=
    	\min_{b(\mathcal{Q})} \left\{\int_\mathcal{Q} (\mathcal{V}_{P}(X-b(P)) - \mathbb{E}_{P}[X-b(P)+b(P)] )\mathrm{d} \mu_\mathcal{Q} (P) \mid \int_\mathcal{Q} b(P) \mathrm{d} \mu_\mathcal{Q}(P) = 0 \right\} \\
    	&=
    	\min_{b(\mathcal{Q})} \left\{\int_\mathcal{Q} (\mathcal{V}_{P}(X-b(P)) - \mathbb{E}_{P}[X-b(P)] )\mathrm{d} \mu_\mathcal{Q} (P) - \int_\mathcal{Q} b(P) \mathrm{d} \mu_\mathcal{Q} (P)\mid \int_\mathcal{Q} b(P) \mathrm{d} \mu_\mathcal{Q}(P) = 0 \right\} \\
    	&= 
    	\min_{b(\mathcal{Q})} \left\{\int_\mathcal{Q} \mathcal{E}_{P}(X-b(P))  \mathrm{d} \mu_\mathcal{Q} (P) \mid \int_\mathcal{Q} b(P) \mathrm{d} \mu_\mathcal{Q}(P) = 0 \right\} \\
    	&=\mathcal{E}_{\mathcal{Q}} (X).
    \end{align*}

Then we move on Eqs.(\ref{Erelation3}) and (\ref{Erelation4}). We observe that
\begin{align*}
	c + \mathcal{V}_{\mathcal{Q}} (X-c) 
	&=
	c + \min_{b(\mathcal{Q})} \left\{\int_\mathcal{Q} \mathcal{V}_{P}(X-c-b(P)) \mathrm{d} \mu_\mathcal{Q} (P) \mid \int_\mathcal{Q} b(P) \mathrm{d} \mu_\mathcal{Q}(P) = 0 \right\}\\
	&=
	\min_{b(\mathcal{Q})} \left\{\int_\mathcal{Q} \left(\mathcal{V}_{P}(X-c-b(P)) + c + b(P)\right) \mathrm{d} \mu_\mathcal{Q} (P) \mid \int_\mathcal{Q} b(P) \mathrm{d} \mu_\mathcal{Q}(P) = 0 \right\}.
\end{align*}
For each $P \in \mathcal{Q}$, the term $\mathcal{V}_{P}(X-c-b(P)) + c + b(P)$ attains its minimum if and only if $c + b(P) = \mathcal{S}_{P}(X)$, where $\mathcal{S}_{P}(X)$ denotes the statistic of $X$ under probability measure $P$. By the monotonicity of the integral operation under a fixed $\mu_{\mathcal{Q}}$, we have
\begin{align*}
	c + \mathcal{V}_{\mathcal{Q}} (X-c) 
	& \geq 
	\int_\mathcal{Q} \left(\mathcal{V}_{P}(X-\mathcal{S}_{P}(X)) + \mathcal{S}_{P}(X)\right)\mathrm{d} \mu_{\mathcal{Q}} (P) \\
	& =
	\int_\mathcal{Q} \mathcal{R}_{P}(X) \mathrm{d} \mu_{\mathcal{Q}} (P)
	=
	\mathcal{R}_{\mathcal{Q}}(X),
\end{align*}
with equality if and only if $c + b(P) = \mathcal{S}_{P}(X)$ for all $P \in \mathcal{Q}$. Combining this with the constraint $\int_\mathcal{Q} b(P) \mathrm{d} \mu_\mathcal{Q}(P) = 0$, we immediately obtain:
$$
c = \int_\mathcal{Q} (c+b(P)) \mathrm{d} \mu_\mathcal{Q}(P) = \int_\mathcal{Q} \mathcal{S}_{P}(X) \mathrm{d} \mu_\mathcal{Q}(P),
$$
which further implies that $\arg\min_{c} \left\{ c + \mathcal{V}_{\mathcal{Q}} (X - c) \right\} = \mathcal{S}_{\mathcal{Q}} (X)$.

Similarly, for the $\mathcal{E}_{\mathcal{Q}}(X-c)$ term, we have
\begin{align*}
	\mathcal{E}_{\mathcal{Q}}(X-c)
	=
	\min_{b(\mathcal{Q})} \left\{\int_\mathcal{Q} \mathcal{E}_{P}(X-c-b(P)) \mathrm{d} \mu_\mathcal{Q} (P) \mid \int_\mathcal{Q} b(P) \mathrm{d} \mu_\mathcal{Q}(P) = 0\right\}.
\end{align*}
By the same logic, for each $P \in \mathcal{Q}$, $\mathcal{E}_{P}(X-c-b(P))$ attains its minimum if and only if $c + b(P) = \mathcal{S}_{P}(X)$. From this, we can immediately derive that $\mathcal{E}_{\mathcal{Q}}(X-c) \geq \mathcal{D}_{\mathcal{Q}}(X)$, where equality holds if and only if $c + B(P) = \mathcal{S}_{P}(X)$ for all $P \in \mathcal{Q}$; this also implies $\arg\min_{c} \left\{ \mathcal{E}_{\mathcal{Q}} (X - c) \right\} = \mathcal{S}_{\mathcal{Q}} (X)$.

Next, if each paired $(\mathcal{R}_{P}, \mathcal{D}_{P}, \mathcal{V}_{P}, \mathcal{E}_{P})$ is regular, then it is straightforward to check that the corresponding $(\mathcal{R}_{\mathcal{Q}}, \mathcal{D}_{\mathcal{Q}}, \mathcal{V}_{\mathcal{Q}}, \mathcal{E}_{\mathcal{Q}})$ is also regular.

\end{proof}

\end{document}